\documentclass{rQUF2e_preprint}
\usepackage{amsmath}
\usepackage{amsthm}
\usepackage{amssymb}
\usepackage{amsfonts}
\usepackage{subfigure}
\usepackage{xcolor}
\definecolor{indigo}{RGB}{88,86,214}
\usepackage[breaklinks,colorlinks=true,allcolors=indigo]{hyperref}

\theoremstyle{definition}
\newtheorem{theorem}{Theorem}[section]

\newtheorem{proposition}[theorem]{Proposition}
\newtheorem{example}[theorem]{Example}
\theoremstyle{definition}
\newtheorem{definition}[theorem]{Definition}
\theoremstyle{remark}
\newtheorem{remark}[theorem]{Remark}

\newcommand\bbr[1]{\left[{#1}\right]}
\newcommand\pbr[1]{\left({#1}\right)}
\newcommand\buy{y^{\text{\rm{buy}}}}
\newcommand\sell{y^{\text{\rm{sell}}}}
\newcommand\bl{b_{\text{l}}}
\newcommand\bu{b_{\text{u}}}
\newcommand\NN{\mathsf{NN}}
\newcommand\clamp{\mathop{\mathsf{clamp}}}

\begin{document}

\title{No-Transaction Band Network: A Neural Network Architecture for Efficient Deep Hedging}
\author{
    Shota~Imaki$^{\ast\dag}$,
    Kentaro~Imajo$^{\ddag}$,
    Katsuya~Ito$^{\ddag}$,
    Kentaro~Minami$^{\ddag}$ and
    Kei~Nakagawa$^{\mathsection}$%
    \thanks{%
        $\ast$ SI is a corresponding author.
        Emails: shota.imaki.0801@gmail.com,
        imos@preferred.jp,
        katsuya1ito@preferred.jp,
        minami@preferred.jp,
        kei.nak.0315@gmail.com
    }%
    \affil{%
        $^{\dag}$Department of Physics, The University of Tokyo, Tokyo 113-0033, Japan\\
        $^{\ddag}$Preferred Networks, Inc., Otemachi Bldg., 1-6-1 Otemachi, Chiyoda-ku, Tokyo 100-0004, Japan\\
        $^{\mathsection}$Nomura Asset Management Ltd., 2-2-1, Toyosu, Koto-ku, Tokyo 135-0061, Japan
    }%
}
\date{\today}

\maketitle

\begin{abstract}
    Deep hedging~\citep{buehler2019deep} is a versatile framework to compute the optimal hedging strategy of derivatives in incomplete markets.
    However, this optimal strategy is hard to train due to action dependence, that is, the appropriate hedging action at the next step depends on the current action.
    To overcome this issue, we leverage the idea of a no-transaction band strategy, which is an existing technique that gives optimal hedging strategies for European options and the exponential utility.
    We theoretically prove that this strategy is also optimal for a wider class of utilities and derivatives including exotics.
    Based on this result, we propose a \textit{no-transaction band network}, a neural network architecture that facilitates fast training and precise evaluation of the optimal hedging strategy.
    We experimentally demonstrate that for European and lookback options, our architecture quickly attains a better hedging strategy in comparison to a standard feed-forward network.
\end{abstract}


\section{Introduction} \label{sec-introduction}

Hedging derivatives is an essential operation in the trillion-dollar industry of derivatives.
In an idealized world with a complete market (i.e., a market without frictions), the perfect hedge is accessible with Black--Scholes' replicating portfolio~\citep{black1973pricing}.
However, a paradoxical situation known as ``Catch-22 of derivative pricing''~\citep{davis1993european} arises in such a world, which suggests a practical limitation of Black--Scholes' idealized setting:
If investors could replicate derivatives with underlying instruments, these derivatives are merely redundant.
The real market, in contrast, always involves a rich variety of frictions, transaction cost in particular~\citep{collins1991methodology}, and thereby makes hedging optimization much more challenging.

Studies of derivatives under transaction cost date back to \cite{leland1985option}.
This work extended Black--Scholes' replicating portfolio to a discretely rebalanced portfolio, and then \cite{grannan1996minimizing} further pushed their idea to approximate the optimal hedge in the sense of mean-variance.
The concept of ``optimal'' hedge was later refined in the framework of utility indifference pricing~\citep{hodges1989optimal,henderson2004utility}.
Based on this framework, a number of studies have pursued the optimal hedge using Markov chain approximation~\citep{monoyios2004option}, asymptotic expansion with respect to a small transaction cost~\citep{whalley1997asymptotic}, and so forth.
However, it is still challenging to pursue a true optimal hedge while incorporating the nonlinear effects of transaction costs.

Deep hedging~\citep{buehler2019deep} ingeniously broke new ground by representing a hedging strategy by a deep neural network.
By virtue of the high expressivity of a neural network~\citep{hornik1991approximation} and modern optimization algorithms~\citep{kingma2014adam},
one can expect to achieve the optimal hedge by training a neural network.
Indeed, their experiments show high feasibility and scalability of deep hedging algorithms for options under transaction costs.

Nevertheless, training in deep hedging bears an essential difficulty of action-dependence, which is the central issue in this work.
That is, since an appropriate hedging action at the next time step depends on the current action, a neural network must explore a vast function space to find the optimum.
Indeed, as we experimentally show in this paper, a naively designed neural network does not converge to the optimum even after a considerable amount of gradient updates.
Therefore, in this work, we develop an efficient methodology to facilitate training in order to fully leverage the versatility of deep hedging.

Proper inductive biases, namely, task-specific structures incorporated in neural networks, often bring about breakthroughs in deep learning.
Deep neural networks with proper inductive biases can quickly attain better optima,
as exemplified by convolutional neural networks for image recognition~\citep{krizhevsky2017imagenet},
long short term memory for time-series processing~\citep{hochreiter1997long}, and
Transformers for sequence transduction~\citep{vaswani2017attention}.
Considering that there is no free lunch method that acclimates to arbitrary tasks~\citep{wolpert1997no}, it is quite important to design neural networks that leverage task-specific structures.

In our work, what guides us to a proper inductive bias are venerable studies of the ``\textit{no-transaction band strategy}''~\citep{davis1990portfolio,davis1993european}.
This strategy saves on transaction cost by refraining from rebalancing unless one's hedge ratio exceeds some permissible range (i.e., no-transaction band) and is proved to be optimal for European options and the exponential utility~\citep{davis1993european}.
Since the no-transaction band is independent of the current hedge ratio, this strategy alleviates the difficulty of action-dependence and has helped to pursue the optimal hedge in previous studies~\citep{whalley1997asymptotic,kallsen2015option}.

This paper proposes a ``\textit{no-transaction band network}'', which is a neural network architecture for efficient deep hedging.
This architecture has two distinctive features compared to prior works.
First, the neural network does not use the current action (i.e., hedge ratio) but uses any other information a human trader might use as an input.
This feature circumvents the complication of action-dependence and facilitates training.
Second, the neural network outputs a no-transaction band strategy.
We theoretically prove that this strategy is optimal not only for European options but for a wider class of utilities and derivatives including exotics.
This result suggests that we can facilitate training by introducing a neural network architecture that outputs no-transaction band strategies.
In practice, this architecture can be implemented by simply adding one special layer to an arbitrary neural network, which is outlined in Figure~\ref{fig-nn-ntb}.
We experimentally show that our network can quickly make a better hedge and evaluate precise prices of derivatives including exotic options.

\subsection{Related Work}

There is a vast literature on pricing and hedging with neural networks~\citep{ruf2020neural}.
\cite{hutchinson1994nonparametric} first proposed using a neural network as a nonparametric model for option pricing and hedging.
\cite{hutchinson1994nonparametric} demonstrated that neural networks can be used successfully to estimate an option pricing formula, with good out-of-sample pricing and delta-hedging performance.
According to \cite{ruf2020neural}, while there are many papers on pricing using neural networks, only a few papers discuss hedging with neural networks.
Most studies support the effectiveness of hedging with a neural network.
\cite{amilon2003neural} examined whether a neural network can price and hedge a call option better than the Black--Scholes model with historical and implied volatility as a benchmark.
They empirically revealed that neural network models outperform the benchmarks both in pricing and hedging performances.
Recently, \cite{buehler2019deep} developed deep hedging, a neural network framework for hedging under market frictions with convex risk measure as a loss function.
Numerical experiments in \cite{buehler2019deep} demonstrated that a deep hedging algorithm can feasibly be used to hedge a European option under exponential utility.
Our paper extends the deep hedging framework to realize the fast and precise computation of an optimal hedging framework for more general convex risk measures and options.

Also in the field of finance, proper inductive biases often improve the empirical performance of neural networks.
\cite{garcia2000pricing} is one of the first papers that embed financial inductive bias into the construction of neural networks.
In options markets, options are usually analyzed and priced based on the moneyness (i.e., the ratio of the stock price of the underlying asset to the strike price of the option) and time-to-maturity of the option.
They trained their neural network model by what they call the homogeneity hint technique which categorizes and clusters the options data based on the moneyness and time-to-maturity of the option.
Based on the inductive bias, the hint guides the neural network model towards learning about properties of the unknown function.
\cite{das2017new} proposed a hybrid model that combines parametric option pricing models such as the Black--Scholes model, the Monte Carlo option pricing model, and the finite difference method with nonparametric machine learning techniques and incorporate a homogeneity hint into the model.
They confirmed that the proposed hybrid model is viable while effective and provides better predictive performance compared to several benchmark models.

Besides \cite{garcia2000pricing} and \cite{das2017new}, there are several related applications of neural networks with financial inductive bias.
\cite{dugas2009incorporating} first design a neural network architecture that incorporates no-arbitrage conditions such as convexity of option prices.
They proved universal approximation properties for the neural network architecture and demonstrated that these neural networks reduce the bias and variance of the Call option pricing.
\cite{jang2019generative} proposed the generative Bayesian neural networks consistent with the no-arbitrage pricing structure.
The generative Bayesian neural network models obtained a better calibration and prediction performance for American option pricing than classical American option models.
\cite{imajo2020deep} proposed a neural network architecture that incorporates widely accepted financial inductive biases of stock time series such as amplitude and time-scale invariance.
Through empirical studies on U.S. and Japanese stock market data, they demonstrated that their proposed method can significantly improve the performance of portfolios on various performance measures.

In our work, we focus on the no-transaction band strategy as an inductive bias.
When the costs are linear in the dollar amount of stock traded, the optimal strategy is to make no trades when the portfolio weights lie in a certain no-trade interval.
When the portfolio weights reach either endpoint of this interval, the optimal strategy is to trade to keep the weight on the boundary of, or inside, the interval~\citep{constantinides1986capital,davis1990portfolio,dumas1991exact,davis1993european}.
This strategy is optimal in hedging derivatives under a linear transaction cost.
In fact, \cite{davis1993european} proved that this strategy is also an optimal hedging policy for European options and the exponential utility.
Based on the above researches, we propose a neural network architecture that incorporates the no-transaction band that is not used as an inductive bias for network structures.
We experimentally demonstrate that our architecture quickly attains a better hedging strategy in comparison to a standard feed-forward network.

\subsection{Outline}

This paper is organized as follows.
We first quickly review utility indifference pricing in Section~\ref{sec-price}.
Utility indifference pricing is a natural and versatile framework to characterize the optimality of hedging strategy and the fair price of derivatives under market frictions.
Section~\ref{sec-deephedging} summarizes how a deep hedging algorithm seeks the optimal hedging strategy using a neural network.
We also point out practical difficulties in training.
In Section~\ref{sec-ntb}, we propose the no-transaction band network.
After illustrating the implementation and interpretation of this architecture, we show the optimality of the strategy encoded in this network.
In Section~\ref{sec-experiment}, we exhibit the results of numerical experiments.
It is demonstrated that the no-transaction band network outperforms a standard feed-forward network in the fast and precise computation of hedging strategy and derivative prices under transaction cost.
Section~\ref{sec-conclusion} presents our summary and outlook.

\section{Utility Indifference Pricing of Derivatives} \label{sec-price}

We first review the utility indifference pricing of derivatives.
We explain the asset price dynamics, transaction cost modeling, and derivatives in Section~\ref{subsec-market}.
Then, Section~\ref{subsec-criteria} formulates a criterion of a market's participant by means of utility and corresponding convex risk measure.
Finally, in Section~\ref{subsec-hedge}, we define the utility indifference price of derivatives upon optimal hedge in terms of the risk measure.
A more comprehensive formulation can be found in \cite{hodges1989optimal} and \cite{henderson2004utility}.

\subsection{Market} \label{subsec-market}

We consider a market with discrete time steps $t \in \{t_0 = 0, t_1, \dots, t_n = T\}$ and a single tradable asset.
This asset has a mid-price given by a stochastic process $S = \{S_{t_i} | S_{t_i} > 0\}_{0 \leq {t_i} \leq T}$.
One example of a stochastic process is geometric Brownian motion with volatility $\sigma > 0$ and vanishing drift,
\begin{align}
    \Delta S_{t_i}
        & = \sigma S_{t_i} \Delta W_{t_i} \,,
    \label{eq-brown}
\end{align}
where $\Delta S_{t_i} \equiv S_{t_{i + 1}} - S_{t_i}$ and $\Delta W_{t_i}$ is independent normal distribution.

A participant of this market can trade the asset $S$.
We let $\delta = \{\delta_{t_i} | \delta_{t_i} \in \mathbf{R}\}_{0 \leq {t_i} \leq T}$ denote a signed number of shares held by the participant at each time step and suppose $\delta_0 = 0$.
Causality requires the participant to determine $\delta_{t_i}$ based on the information available before ${t_i}$.

The market levies a transaction cost proportional to traded values with a cost rate $c \in [0, 1)$.
The participant's total expense on transaction costs until $t_i$ is denoted by $C_{t_i}(S, \delta)$, which satisfies $C_0(S, \delta) = 0$ and
\begin{align}
    \Delta C_{t_i} (S, \delta)
        & = c S_{t_i}|\Delta \delta_{t_i}| \,.
\end{align}
In practical contexts, the transaction cost includes the bid-ask spread, commission fee, temporary market impact, and so forth~\citep{collins1991methodology}.

A derivative $Z$ of the asset gives a payoff $Z(S)$, which is contingent on a price trajectory.
We assume that the payoff is settled at maturity $t = T$ and discounted by zero risk-free-rate.
Our consideration excludes derivatives allowing for early exercises, such as American options.
The following derivatives are used in our numerical experiments in Section~\ref{sec-experiment}.
\begin{example}[European Option]  \label{example-european}
    A payoff of a European call option with a strike $K$ is given by $Z = \max(S_T - K,  0)$.
\end{example}
\begin{example}[Lookback Option]  \label{example-lookback}
    A payoff of a lookback call option with a fixed strike $K$ is given by $Z = \max[\max(\{S_{t_i}\}_{0 \leq {t_i} \leq T}) - K, 0]$.
\end{example}

\subsection{Criteria}  \label{subsec-criteria}

Let us then describe the preference of a market participant to their portfolio.
The formulation begins with utility.
\begin{definition}[Utility]
    A function $u: \mathbf{R} \to \mathbf{R}$ is called utility if the following conditions are satisfied.
    \begin{enumerate}[(i)]
        \item
            Non-decreasing:
            For $x_1 \leq x_2$, $u(x_1) \leq u(x_2)$ holds.
        \item
            Concavity:
            For $a \in [0, 1]$, $u(a x_1 + (1 - a) x_2) \geq a u(x_1) + (1 - a) u(x_2)$ holds.
    \end{enumerate}
\end{definition}

An optimized certainty equivalent~\citep{ben1986expected,ben2007old} is a convex risk measure associated with the utility.
\begin{definition}[Optimized Certainty Equivalent]
    Let $u$ be a utility function.
    For a real-valued random variable $X$, we define an optimized certainty equivalent $\rho_u$ corresponding to $u$ as follows.
    \begin{align}
        \rho_u(X)
            & \equiv
                \inf_{w \in \mathbf{R}}
                \{w - \mathbf{E}[u(X + w)]\} \,.
        \label{eq-rho}
    \end{align}
\end{definition}
\begin{proposition}[Optimized Certainty Equivalent as Convex Risk Measure]
    An optimized certainty equivalent is a convex risk measure, namely, satisfies the following properties.
    \begin{enumerate}[(i)]
        \item
            Non-increasing:
            For $X_1 \leq X_2$,
            $\rho_u(X_1) \geq \rho_u(X_2)$ holds.
        \item
            Convexity:
            For $a \in [0, 1]$,
            $\rho_u(a X_1 + (1 - a) X_2) \leq a \rho_u(X_1) + (1 - a) \rho_u(X_2)$ holds.
        \item
            Cash equivalent:
            For $\eta \in \mathbf{R}$, $\rho_u(X + \eta) = \rho_u(X) - \eta$ holds.
    \end{enumerate}
\end{proposition}
\begin{proof}
    (i) Since a utility $u$ is non-decreasing, for $X_1 \leq X_2$, $w - \mathbf{E}[u(X_1 + w)] \geq w - \mathbf{E}[u(X_2 + w)]$ holds and therefore it follows that $\rho_u(X_1) \geq \rho_u(X_2)$.
    (ii) Let us denote $X_a \equiv a X_1 + (1 - a) X_2$ and $w_a \equiv a w_1 + (1 - a) w_2$ for $a \in [0, 1]$.
    The concavity of a utility $u$ implies that $\mathbf{E}[w_a - u(X_a + w_a)] \leq a \mathbf{E}[w_1 - u(X_1 + w_1)] + (1 - a) \mathbf{E}[w_2 - u(X_2 + w_2)]$.
    Taking the infima with respect to $w_1$, $w_2$, $w_a$ of the both hand sides yields $\rho_u(X_a) \leq a \rho_u(X_1) + (1 - a) \rho_u(X_2)$.
    (iii) A short algebra gives $\rho_u(X + \eta) = \inf_w \mathbf{E}[(w + \eta) - u(X + w + \eta)] = \inf_{w'} \mathbf{E}[w' - u(X + w')] = \rho_u(X)$ where $w' = w + \eta$.
\end{proof}
In particular, we employ the entropic risk measure in our experiments in Section~\ref{sec-experiment}.
Derivative pricing based on this risk measure gives exponential utility indifference price, which has been studied in \cite{hodges1989optimal}, \cite{davis1993european}, and \cite{whalley1997asymptotic}.
\begin{example}[Entropic Risk Measure]  \label{example-entropic}
    Consider exponential utility, $u(x) = - \exp(-\lambda x)$ with $\lambda > 0$ being a risk aversion coefficient.
    The corresponding optimized certainty equivalent is called entropic risk measure, which reads $\rho_u(X) = (1 / \lambda) \log \mathbf{E} [\exp(-\lambda X)]$.
\end{example}

\subsection{Hedging and Pricing} \label{subsec-hedge}

Before formulating hedging and pricing, it would be informative to illustrate a typical situation in which a dealer sells over-the-counter derivatives to its customers.
A dealer sells a derivative to its customer and consequently obliges a liability to settle the payoff at maturity.
The dealer may hedge the risk of this liability by trading an underlying asset of the derivative.
On the premise of the optimal hedge in terms of the dealer's risk measure, the dealer quotes such a price that compensates the residual risk after hedging.

A dealer of a derivative $Z$ may hedge their liability by their trading strategy $\delta$.
The resulting final wealth at maturity is given by
\begin{align}
    P(-Z, S, \delta)
        & = - Z(S) + (\delta \cdot S)_T - C_T(S, \delta) \,,
        \label{eq-wealth}
    \\
    (\delta \cdot S)_T
        & \equiv \sum_{i = 0}^{n - 1} \delta_{t_i} \Delta S_{t_i} \,,
\end{align}
which adds up the payoff to the customer, capital gains from the underlying asset, and the transaction cost.
Our formulation neglects transaction cost to liquidate the terminal position.
The hedger wishes to optimize the distribution of the final wealth~\eqref{eq-wealth} in terms of their convex risk measure $\rho_u$.
The optimal convex risk measure qualifies as a risk measure and is interpreted as the derivative's residual risk after the optimal hedge.
\begin{definition}[Optimal Convex Risk Measure]
    The optimal convex risk measure is defined as the infimum
    \begin{align}
        \pi_u(-Z)
            & \equiv
                \inf_{\delta}
                \rho_u (P(-Z , S, \delta)) \,.
        \label{eq-pi}
    \end{align}
    The minimizer of the problem~\eqref{eq-pi} is called the optimal hedging strategy for $Z$.
\end{definition}
\begin{proposition}[Optimal Convex Risk Measure as Convex Risk Measure]
    An optimal risk measure~\eqref{eq-pi} qualifies as a convex risk measure, namely, satisfies the following properties.
    \begin{enumerate}[(i)]
        \item
            Non-increasing:
            For $X_1 \leq X_2$, $\pi_u(X_1) \geq \pi_u(X_2)$ holds.
        \item
            Convex:
            For $a \in [0, 1]$, $\pi_u(a X_1 + (1 - a) X_2) \leq a \pi_u(X_1) + (1 - a) \pi_u(X_2)$ holds.
        \item
            Cash equivalent:
            For $\eta \in \mathbf{R}$, $\pi_u(X + \eta) = \pi_u(X) - \eta$ holds.
    \end{enumerate}
\end{proposition}
\begin{proof}
    (i) Since a convex risk measure $\rho_u$ is non-increasing,
    for $X_1 \leq X_2$, $P(X_1, S, \delta) \leq P(X_2, S, \delta)$ and so $\rho_u(P(X_1, S, \delta)) \geq \rho_u(P(X_2, S, \delta))$ holds.
    Therefore it follows that $\pi_u(X_1) \geq \pi_u(X_2)$.
    (ii)
    Let us denote $X_a = a X_1 + (1 - a) X_2$ and $\delta_a = a \delta_1 + (1 - a) \delta_2$ for $a \in [0, 1]$.
    Since $C_T(\delta, S)$ is convex with respect to $\delta$, we have $P(X_a, S, \delta_a) \geq a P(X_1, S, \delta_1) + (1 - a) P(X_2, S, \delta_2)$ and therefore
    the convexity of $\rho_u$ gives rise to $\rho_u(P(X_a, S, \delta_a)) \leq a \rho_u(P(X_1, S, \delta_1)) + (1 - a) \rho_u(P(X_2, S, \delta_2))$.
    Taking the infima with respect to $\delta_a$, $\delta_1$, $\delta_2$ of the both hand sides yields $\pi_u(X_a) \leq a \pi_u(X_1) + (1 - a) \pi_u(X_2)$.
    (iii)
    Since $P(X + \eta, S, \delta) = P(X, S, \delta) + \eta$, the cash equivalence of the convex risk measure implies
    $\pi_u(X + \eta) = \inf_\delta \rho_u(P(X, S, \delta) + \eta) = \inf_\delta \rho_u(P(X, S, \delta)) - \eta = \pi_u(X) - \eta$.
\end{proof}

We are now ready to define the utility indifference price of a derivative $Z$.
The fair price $p$ is identified with such an amount of cash that a hedger, who will optimally hedge their liability, deems indifferent between obliging a liability with $p$ being paid and having no liability.
This condition gives rise to the equation $\pi_u(- Z + p) = \pi_u(0)$, which in turn suggests the following definition.
\begin{definition}[Utility Indifference Price]
    A utility indifference price with respect to utility $u$ of a derivative $Z$ is given by
    \begin{align}
        p_u(Z)
            & \equiv \pi_u(-Z) - \pi_u(0) \,.
        \label{eq-price}
    \end{align}
\end{definition}

The utility indifference pricing extends the pricing via derivative replication in a complete market.
\begin{proposition}[Relation to Replication Pricing in Complete Market]
    Assume a market without transaction cost.
    If $Z$ is replicable by $S$, that is, $P(-Z, S, \delta^*) = -p_0$ holds for some hedging strategy $\delta^*$ and $p_0 \in \mathbf{R}$, the utility indifference price is given by $p_u(Z) = p_0$.
\end{proposition}
\begin{proof}
    For arbitrary strategy $\delta$, we have $P(-Z, S, \delta) = P(-Z, S, \delta^*) + P(0, S, \delta - \delta^*) = -p_0 +  P(0, S, \delta - \delta^*)$ and therefore
    $\rho_u(P(-Z, S, \delta)) = p_0 + \rho_u(P(0, S, \delta - \delta^*))$.
    Taking the infima with respect to $\delta$ of the both hand sides yields $\pi_u(-Z) = p_0 + \pi_u(0)$ and therefore $p_u(Z) = p_0$.
\end{proof}
\begin{remark}[Delta Hedge] \label{remark-deltahedge}
    One can see that the replicating strategy $\delta^*$ in the absence of cost can be approximated by delta hedge, $\delta_t^* \simeq \partial \mathbf{E}[Z | \{S_{t'}\}_{t' < t}] / \partial S_t$,
    if $\Delta t$ is small.  The limit $\Delta t \to 0$ would give rise to the continuous delta hedge and accordingly Black--Scholes price.
\end{remark}

\section{Review of Deep Hedging}  \label{sec-deephedging}

The purpose of this section is to review deep hedging~\citep{buehler2019deep}.
The central idea of deep hedging is to represent a hedging strategy by a deep neural network.
The universal approximation theorem~\citep{hornik1991approximation} implies that a neural network can in principle approximate the optimal strategy arbitrarily well, and sophisticated optimization algorithms train a neural network toward the optimum.
Nonetheless, training in practice has a complexity of action-dependence.

A feed-forward neural network comprises layers of neurons with nonlinear activation.
More formally, with $n_0, n_1, \dots, n_L$ being the number of neurons in each layer and $W_l$ being an affine transformation $\mathbf{R}^{n_{l - 1}} \to \mathbf{R}^{n_{l}}$,
a neural network $\NN: \mathbf{R}^{n_0} \to \mathbf{R}^{n_L}$ is given by
\begin{align}
    \NN
        & = W_L \circ \sigma \circ W_{L - 1} \circ \cdots \circ \sigma \circ W_1 \,.
\end{align}
Here $\sigma$ is an activation function applied elementwise.

Deep hedging algorithm represents a hedging strategy $\delta$ by a neural network.
The architecture proposed in \cite{buehler2019deep} uses the current hedge ratio as an input to a neural network.
Namely, the hedging strategy is given by the following semi-recurrent neural network.
\begin{align}
    \delta_{t_{i + 1}}
        & = \NN(I_{t_i}, \delta_{t_i}) \,.
    \label{eq-nn-ff}
\end{align}
A schematic diagram is depicted in Figure~\ref{fig-nn-ff}.
Here $I_{t_i}$ denotes a set of relevant information available at $t_i$.
For instance, for a European option of an asset with the price following Markovian process, $I_{t_i} = \{S_{t_i}, t_i\}$ would suffice.
One can expect that the neural network representation~\eqref{eq-nn-ff} with sufficient parameters includes the optimal hedging strategy because the universal approximation theorem~\citep{hornik1991approximation} implies that a neural network can approximate a certain class of functions arbitrarily well.

The neural network representation~\eqref{eq-nn-ff} translates the hedging and pricing problem to the optimization of a set of parameters $\theta$ of $\NN$ toward the infimal convex risk measure~\eqref{eq-pi}.
One numerical approach to obtain an approximate infimum is the gradient descent method, which in its simplest form iteratively improves the parameters as follows.
\begin{align}
    \theta \to \theta - \eta \frac{\partial \rho_u}{\partial \theta} \,,
\end{align}
where $\eta$ is a positive constant.
More refined optimization algorithms such as \href{https://pytorch.org/docs/stable/optim.html?highlight=adam#torch.optim.Adam}{$\mathsf{Adam}$}~\citep{kingma2014adam} have been invented to avoid being trapped in local minima.
After sufficient iterations, one can expect to obtain a close-to-optimal hedging strategy $\delta^*$.
This strategy approximates the optimal convex risk measure~\eqref{eq-pi} as
\begin{align}
    \pi_u(-Z)
        & \simeq \rho_u(P(-Z, S, \delta^*))
\end{align}
and accordingly yields the approximate value of the utility indifference price~\eqref{eq-price}.

Nevertheless, such optimization is easier said than done in practice due to action-dependence.
That is, since the input of the neural network is dependent on $\delta_{t_i}$ as Eq.~\eqref{eq-nn-ff}, a neural network should explore a vast function space to find the optimum.
Indeed, as we experimentally demonstrate in Section~\ref{sec-experiment}, a naive implementation of deep hedging does not converge to the optimal hedging strategy even after a considerable amount of training.

\section{Proposed Architecture: No-Transaction Band Network}  \label{sec-ntb}

Now, we introduce a new neural network architecture for efficient deep hedging: a no-transaction band network.
Section~\ref{subsec-architecture} illustrates the implementation and interpretation of the proposed network.
One can easily build this network by simply adding a layer using \href{https://pytorch.org/docs/stable/generated/torch.clamp.html?highlight=clamp#torch.clamp}{$\clamp$} function which is available in major deep learning libraries.
This network naturally encodes a proper inductive bias for hedging problems and alleviates the complication of action-dependence.
Section~\ref{subsec-rationale} demonstrates that such a strategy is optimal for derivatives with finite delta and gamma.

\subsection{Architecture} \label{subsec-architecture}

The proposed network is as simple as follows.
The neural network outputs a range $[\bl, \bu]$ and the next hedge ratio is obtained by clamping the current hedge ratio into this range.
That is,
\begin{align}
    \begin{aligned}
        (\bl, \bu)
            & = \NN(I_{t_i}) \,,
        \\
        \delta_{t_{i + 1}}
            & = \clamp(\delta_{t_i}, \bl, \bu) \,,
    \end{aligned}
    \label{eq-ntb}
\end{align}
where the \href{https://pytorch.org/docs/stable/generated/torch.clamp.html?highlight=clamp#torch.clamp}{$\clamp$} function is given by
\begin{align}
    \clamp(\delta_{t_i}, \bl, \bu)
        & \equiv
        \begin{cases}
            \bl
                & \text{if $\delta_{t_i} < \bl$} \\
            \delta_{t_i}
                & \text{if $\bl \leq \delta_{t_i} \leq \bu$} \\
            \bu
                & \text{if $\delta_{t_i} > \bu$}
        \end{cases}
        \,.
        \label{eq-clamp}
\end{align}
A schematic diagram is depicted in Figure~\ref{fig-nn-ntb}.
In this strategy, a hedger always maintains their hedge ratio in the range $[\bl, \bu]$ while they never transact inside this range.
This range is called no-transaction band in the literature~\citep{davis1993european,whalley1997asymptotic} and thus we dub this architecture ``\textit{no-transaction band network}''.

The rationale for the advantage of our network is twofold:
First, since the neural network in Eq.~\eqref{eq-ntb} does not need the current hedge ratio $\delta_{t_i}$ as an input, we can circumvent the complexity of action-dependence.
Second, the architecture~\eqref{eq-ntb} incorporates a proper inductive bias because such a strategy is proved to be optimal to hedge derivatives with finite delta and gamma.
Along with the proof given shortly in Section~\ref{subsec-rationale}, we can present an intuitive explanation of optimality:
A hedger should refrain from rebalancing unless the hedge ratio exceeds some permissible range, because, even if the price deviates slightly, at the next moment the price may revert to its earlier value and the transaction cost may be wasted.

As we clarify in the next section~\ref{subsec-rationale}, we believe that the no-transaction band network can give the optimal strategy for any utility and derivatives, as far as utility and payoff function are sufficiently smooth.

We leave two tips to train the no-transaction band network.
First, since the gradient of \href{https://pytorch.org/docs/stable/generated/torch.clamp.html?highlight=clamp#torch.clamp}{$\clamp$} function vanishes in the clamped regions, introducing a small slope in these regions may help it backpropagating losses efficiently.
Second, since the above definition of \href{https://pytorch.org/docs/stable/generated/torch.clamp.html?highlight=clamp#torch.clamp}{$\clamp$}~\eqref{eq-clamp}
(and its implementations in deep learning libraries as well)
does not expect an inverted input $\bl > \bu$, one may wish to give such an input a natural definition such as $\clamp(\delta_{t_i}, \bl, \bu) \equiv (\bl + \bu) / 2$, as we do in our experiments.

\begin{figure}
    \begin{center}
        \begin{minipage}{\linewidth}
            \begin{center}
                \subfigure[Feed-forward network.]{
                    \resizebox*{0.48\linewidth}{!}{\includegraphics{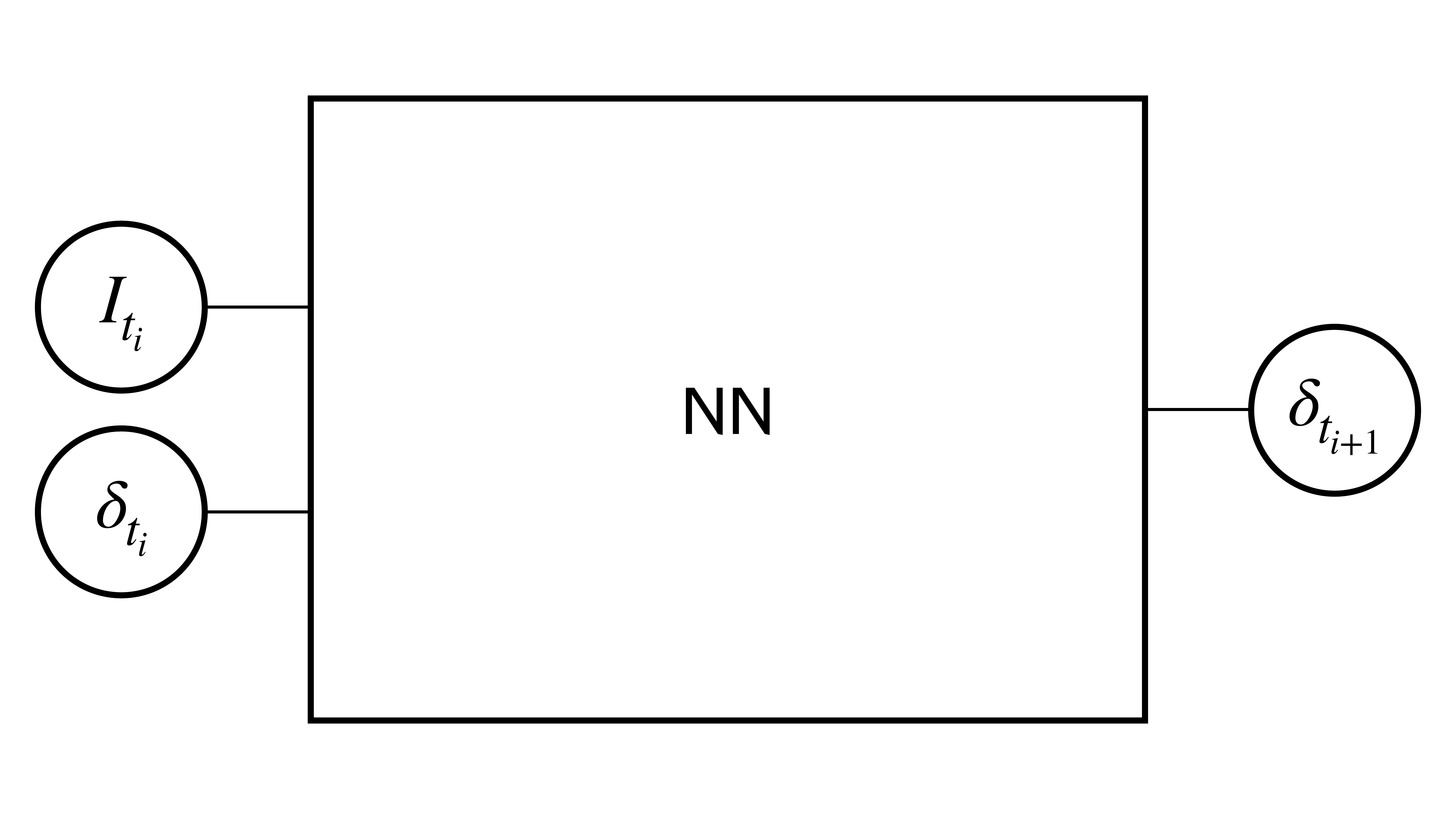}}
                    \label{fig-nn-ff}
                }
                \subfigure[No-transaction band network.]{
                    \resizebox*{0.48\linewidth}{!}{\includegraphics{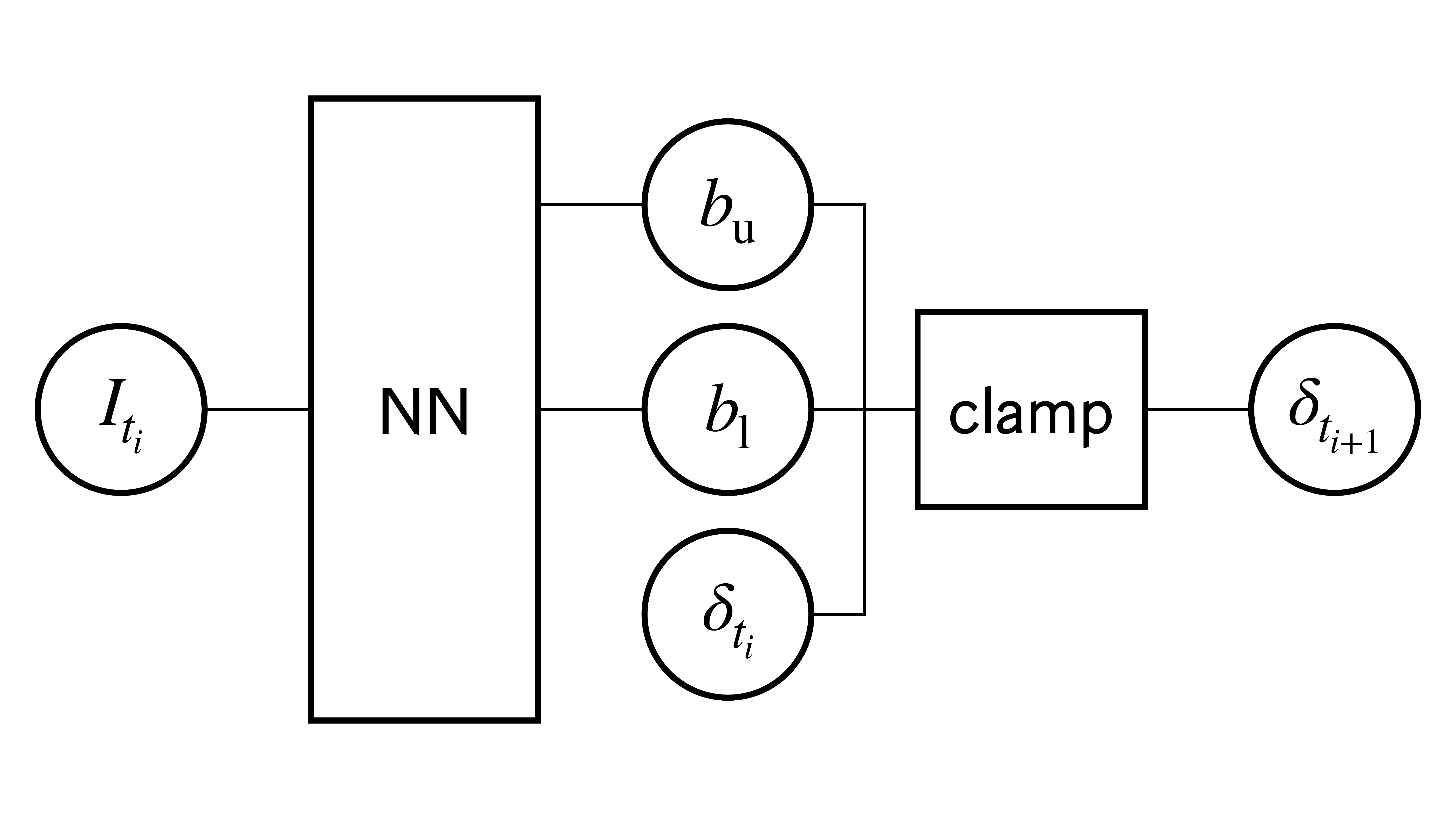}}
                    \label{fig-nn-ntb}
                }
                \caption{
                    Schematic diagrams of an ordinary feed-forward network and the no-transaction band network.
                    Here $I_{t_i}$ denotes a set of relevant information available at $t_i$.
                    The feed-forward network requires the current hedge ratio $\delta_{t_i}$ as an input, whereas the no-transaction band network does not.
                    The feed-forward network directly outputs the next hedge ratio $\delta_{t_{i + 1}}$.
                    On the other hand, the no-transaction band network computes a range $[\bl, \bu]$ by a neural network and then determines the next hedge ratio by clamping $\delta_{t_i}$ into this range.
                }
            \end{center}
        \end{minipage}
    \end{center}
\end{figure}

\subsection{Optimality of No-Transaction Band Strategy}  \label{subsec-rationale}

Let us complete the rationale of our network by demonstrating the optimality of the no-transaction band strategy.
The main statement is Proposition~\ref{proposition-optimality}, which extends the argument of \cite{davis1993european} to more general utility and path-dependent derivatives.

To this end, we define a loss function,
\begin{align}
    \ell_{t_i}(-Z; \{S_{t_j}\}_{t_j \leq t_i}, \delta_{t_i}, C_{t_i}, w)
        & = \inf_{\{\Delta \delta_{t_j}\}_{t_j \geq t_i}}
            \mathbf{E} \left[
                -u(P(-Z, S, \delta) + w)
                \middle|
                \{S_{t_j}\}_{t_j \leq t_i},
                \delta_{t_i},
                C_{t_i}
            \right] \,,
    \label{eq-loss}
\end{align}
where $u$ stands for a twice-differentiable utility function.
One can readily see that the minimizer of the loss function~\eqref{eq-loss} also optimizes the convex risk measure~\eqref{eq-rho}.
The next essential step is to formally write down the Hamilton--Jacobi--Bellman equation for the problem~\eqref{eq-loss}.
For notational convenience, we denote the numbers of shares bought and sold at $t_i$ by $\buy_{t_i} \Delta {t_i} \geq 0$ and $\sell_{t_i} \Delta {t_i} \geq 0$.
Then the desired equation for small $\Delta t_i$ reads~\citep{davis1993european}
\begin{align}
    & \inf_{\buy_t, \sell_t}
        \bbr{
            \pbr{
                \frac{\partial \ell_t}{\partial \delta_t} + c \frac{\partial \ell_t}{\partial C_t}
            } S_t \buy_t
            +
            \pbr{
                - \frac{\partial \ell_t}{\partial \delta_t} + c \frac{\partial \ell_t}{\partial C_t}
            } S_t \sell_t
        }
        + \frac{\partial \ell_t}{\partial t}
        + \frac 1 2 \sigma^2 S_t^2 \frac{\partial^2 {\ell_t}}{\partial S_t^2}
    = 0 \,,
    \label{eq-bellman}
\end{align}
where $t = t_i$ and we neglect higher-order contributions with respect to $\Delta t$.
An assumption behind Eq.~\eqref{eq-bellman} is that the derivative $Z$ has a sufficiently smooth payoff so that the loss function is twice differentiable since otherwise, the higher-order terms are no longer negligible.

Now that Eq.~\eqref{eq-bellman} has boiled our problem down to local optimization, we are ready to prove the main proposition.
Although we limit ourselves to a driftless asset under zero risk-free-rate for brevity, this limitation can be relaxed as in \cite{davis1993european}.
Also, our argument can be extended to the case in which each transaction cost incurs a fixed transaction cost, based on the arguments in \cite{whalley1999optimal} and \cite{zakamouline2006european}.
\begin{proposition}[Optimality of No-Transaction Band Strategy] \label{proposition-optimality}
    Let $u$ be a twice-differentiable utility function.
    Assume $Z$ is a derivative with a sufficiently smooth payoff so that the loss function~\eqref{eq-loss} is twice differentiable with respect to $S_t$.
    Then, the optimal hedge for $Z$ is attained by the no-transaction band strategy.
    Namely, the minimizer $\delta^*$ of the loss function~\eqref{eq-loss} is given by $\delta^*_{t_{i + 1}} = \clamp(\delta^*_{t_i}, \bl, \bu)$ where $\bl$ and $\bu$ are functions of $t$ and $\{S_{t'}\}_{t' \leq t}$ satisfying $\bl \leq \bu$.
\end{proposition}
\begin{proof}
    Bellman's optimization principle implies that the optimal hedge ratio $\delta^*$ is obtained by the minimization problem in Eq.~\eqref{eq-bellman}.  For the moment we restrict $\buy_t, \sell_t \leq k < \infty$.
    Our problem is divided into three cases.
    First, if
    \begin{align}
        \frac{\partial \ell_t}{\partial \delta_t} + c \frac{\partial \ell_t}{\partial C_t}
        < 0 \,,
        \quad
        - \frac{\partial \ell_t}{\partial \delta_t} + c \frac{\partial \ell_t}{\partial C_t}
        > 0 \,,
        \label{case-buy}
    \end{align}
    then the minimum is attained by $\buy_t = k$ and $\sell_t = 0$, namely, the hedger should buy the underlying asset as much as possible.
    Second, if
    \begin{align}
        \frac{\partial \ell_t}{\partial \delta_t} + c \frac{\partial \ell_t}{\partial C_t}
        > 0 \,,
        \quad
        - \frac{\partial \ell_t}{\partial \delta_t} + c \frac{\partial \ell_t}{\partial C_t}
        < 0 \,,
        \label{case-sell}
    \end{align}
    then the minimum is attained by $\buy_t = 0$ and $\sell_t = k$, namely, the hedger should sell the underlying asset as much as possible.
    Third, if
    \begin{align}
        \frac{\partial \ell_t}{\partial \delta_t} + c \frac{\partial \ell_t}{\partial C_t}
        \geq 0 \,,
        \quad
        - \frac{\partial \ell_t}{\partial \delta_t} + c \frac{\partial \ell_t}{\partial C_t}
        \geq 0 \,,
        \label{case-no}
    \end{align}
    then the optimal strategy is to make no transaction, $\buy_t = \sell_t = 0$.
    Considering the inequality
    \begin{align}
        \pbr{
            \frac{\partial \ell_t}{\partial \delta_t} + c \frac{\partial \ell_t}{\partial C_t}
        }
        + \pbr{
            - \frac{\partial \ell_t}{\partial \delta_t} + c \frac{\partial \ell_t}{\partial C_t}
        }
        = 2 c \frac{\partial \ell_t}{\partial C_t} \geq 0 \,,
        \label{eq-coefadd}
    \end{align}
    the above cases are collectively exhaustive.
    Let us assume that the coefficient $\partial \ell_t / \partial \delta_t + c \partial \ell_t / \partial C_t$ is monotone increasing and
    $- \partial \ell_t / \partial \delta_t + c \partial \ell_t / \partial C_t$ is monotone decreasing with respect to $\delta_t$.
    This would hold for small $c$ because the loss function~\eqref{eq-loss} is convex with respect to $\delta_t$.
    Then the equations $\pm \partial \ell_t / \partial \delta + c \partial \ell_t / \partial C = 0$ have only one root each, which we denote by $\delta_t = \bl, \bu$ respectively.
    One can see that $\bl \leq \bu$.
    Taking the limit $k \to \infty$, we find that the optimal strategy is given by
    $\delta_{t + \Delta t} = \delta_{t} + \buy_t \Delta t - \sell_t \Delta t = \clamp(\delta_t, \bl, \bu)$, or equivalently $\delta_{t_{i + 1}} = \clamp(\delta_{t_i}, \bl, \bu)$.
\end{proof}

\begin{remark}[Asymptotically Optimal Strategy for Small Cost] \label{whalleywilmott}
    The optimal no-transaction band for exponential utility and asymptotically small transaction cost is evaluated by \cite{whalley1997asymptotic}.
    The result reads
    \begin{align}
        \Delta - \bl
            = \bu - \Delta
            = \pbr{\frac{3 c S_t \Gamma^2}{2 \lambda}}^{1 / 3}
            + o(c^{1 / 3}) \,,
        \label{eq-ww}
    \end{align}
    where $\Delta$ and $\Gamma$ are Black--Scholes' delta and gamma of the derivative.
    Although they consider a European call option, we expect that this result applies to other derivatives for which loss function is smooth.
    We will employ this result as a benchmark of the numerical experiments.
\end{remark}

\section{Experimental Results}  \label{sec-experiment}

We now experimentally demonstrate the efficiency of the proposed architecture.
The attained utility, price, and training speed are compared to an ordinary feed-forward network.
The results show that the no-transaction band network attains faster convergence to higher utility and cheaper prices for a wide range of transaction costs.

Throughout the experiments, we consider the European option and lookback option defined in Section~\ref{subsec-market}.
We consider a European option in Section~\ref{subsec-eu} and a lookback option in Section~\ref{subsec-lb}.
A European option provides the simplest benchmark whereas a lookback option illustrates that our architecture applies to an exotic option.
The price of their underlying asset follows geometric Brownian motion~\eqref{eq-brown} with the initial price $S_0 = 1$ and the volatility $\sigma = 20 \, \%$.
Maturity of the derivatives is $T = 30 / 365$ and time steps are given by $t = 0, 1 / 365, \dots, 30 / 365$.
We compute the prices of the derivatives under 22 transaction costs, $c = 0$ and $c = \exp(k)$ with $k = -10.00, -9.75, \dots, -5.00$.

We compare the following three hedging methods.
(i) \textit{No-transaction band network}:
The architecture is described in Eq.~\eqref{eq-ntb} and Figure.~\ref{fig-nn-ntb}.
We use a four-layer feed-forward neural network with 32 hidden neurons in each layer.
We employ \href{https://pytorch.org/docs/stable/generated/torch.nn.ReLU.html#torch.nn.ReLU}{$\mathsf{ReLU}$} as an activation function.
Two outputs of the neural network are applied with \href{https://pytorch.org/docs/stable/generated/torch.nn.LeakyReLU.html#torch.nn.LeakyReLU}{$\mathsf{LeakyReLU}$} with negative slope $0.01$
and then added and subtracted to the Black--Scholes' delta to get $\bu$ and $\bl$, respectively.
Input features $I_{t_i}$ depend on derivatives and are explained later.
(ii) \textit{Feed-forward neural network}:
The architecture is described in Eq.~\eqref{eq-nn-ff} and Figure.~\ref{fig-nn-ff}.
The neural network comprises $4 \times 32$ neurons and \href{https://pytorch.org/docs/stable/generated/torch.nn.ReLU.html#torch.nn.ReLU}{$\mathsf{ReLU}$} activations.
\href{https://pytorch.org/docs/stable/generated/torch.nn.Tanh.html#torch.nn.Tanh}{$\mathsf{Tanh}$} is applied to the output and then added to the Black--Scholes' delta to get the hedge ratio $\delta_{t_i}$.
(iii) \textit{Asymptotic solution for small transaction cost}:
Details are described in Remark~\ref{whalleywilmott}.
This method is a benchmark for small $c$ because it is supposed to give the optimal hedge under the limits $c \to 0$ and $\Delta t \to 0$.

Neural networks are trained as follows.
In each simulation, we generate 50,000 Monte Carlo paths of the asset prices.
The loss function is computed as the negative of an expected exponential utility with $\lambda = 1$.
We employ the \href{https://pytorch.org/docs/stable/optim.html?highlight=adam#torch.optim.Adam}{$\mathsf{Adam}$}~\citep{kingma2014adam}
optimizer with the parameters recommended in the original paper.
We carry out 1,000 iterations of Monte Carlo simulations for each tuple of option, method, and transaction cost.
We run all the experiments using \href{https://pytorch.org/}{PyTorch}~\citep{NEURIPS2019_9015} and will release the codes\footnote{GitHub: \href{https://github.com/pfnet-research/notransactionbandnetwork}{pfnet-research/NoTransactionBandNetwork}, \href{https://github.com/pfnet-research/pfhedge}{pfnet-research/pfhedge}.}.

\subsection{European Option} \label{subsec-eu}

We consider a European call option with a strike $K = 1$ (See Example~\ref{example-european}).
The set of relevant information at time $t$ comprises log-moneyness, time to maturity, and volatility, $I_{t_i} = \{\log(S_{t_i} / K), T - t_i, \sigma\}$.

As presented in Figure~\ref{fig-eu-utility}, the proposed architecture attains the highest expected utility for a wide range of transaction costs.
It is worth mentioning that the utility of the no-transaction band network ceases to decrease at some value of cost.
This is because the no-transaction band learns ``not to hedge'' when a transaction cost is too expensive.
This advantage is distinctive from the other two methods.

Accordingly, the no-transaction band network allows for cheaper prices.
Figure~\ref{fig-eu-price} shows price spreads $p(c) - p(0)$ attained by each method.
The price at zero cost $p(0)$ is set to the common value obtained by delta-hedge with Black--Scholes' delta (See Remark~\ref{remark-deltahedge}).
One can see that the no-transaction band network attains the cheapest prices for diverse transaction costs.
It is interesting to note that the renowned scaling law, $p(c) - p(0) \propto c^{2/3}$ for $c \to 0$ \citep{whalley1997asymptotic}, is approximately observed for the asymptotic solution.
Also, the asymptotic proximity to this scaling behavior suggests the optimality of the no-transaction band network for small $c$.
See Appendix.~\ref{appendix-table} for the obtained values of expected utility and prices.

Figure~\ref{fig-eu-band} shows the no-transaction band as a function of log-moneyness.
For small transaction costs, the no-transaction band of our network is close to the asymptotic solution but a small difference is observed.
In contrast, when the transaction cost is too expensive, our model outputs an almost vanishing hedge ratio.
In other words, a neural network learns not to hedge to save on transaction costs.

Last but not least, the no-transaction band network has another virtue of fast learning.
The learning histories in Figure~\ref{fig-eu-history} demonstrate that the no-transaction band network reaches its optima in less than 100 simulations, even though the ordinary feed-forward network does not converge in 1,000 simulations.
Varying the learning rate of the optimizer between 0.01 and 0.0001 does not affect the advantage of the no-transaction band network (See Appendix.~\ref{appendix-lr}).

\begin{figure}
    \begin{center}
        \begin{minipage}{\linewidth}
            \begin{center}
                \subfigure[Utility for European option.]{
                    \resizebox*{0.48\linewidth}{!}{
                        \includegraphics{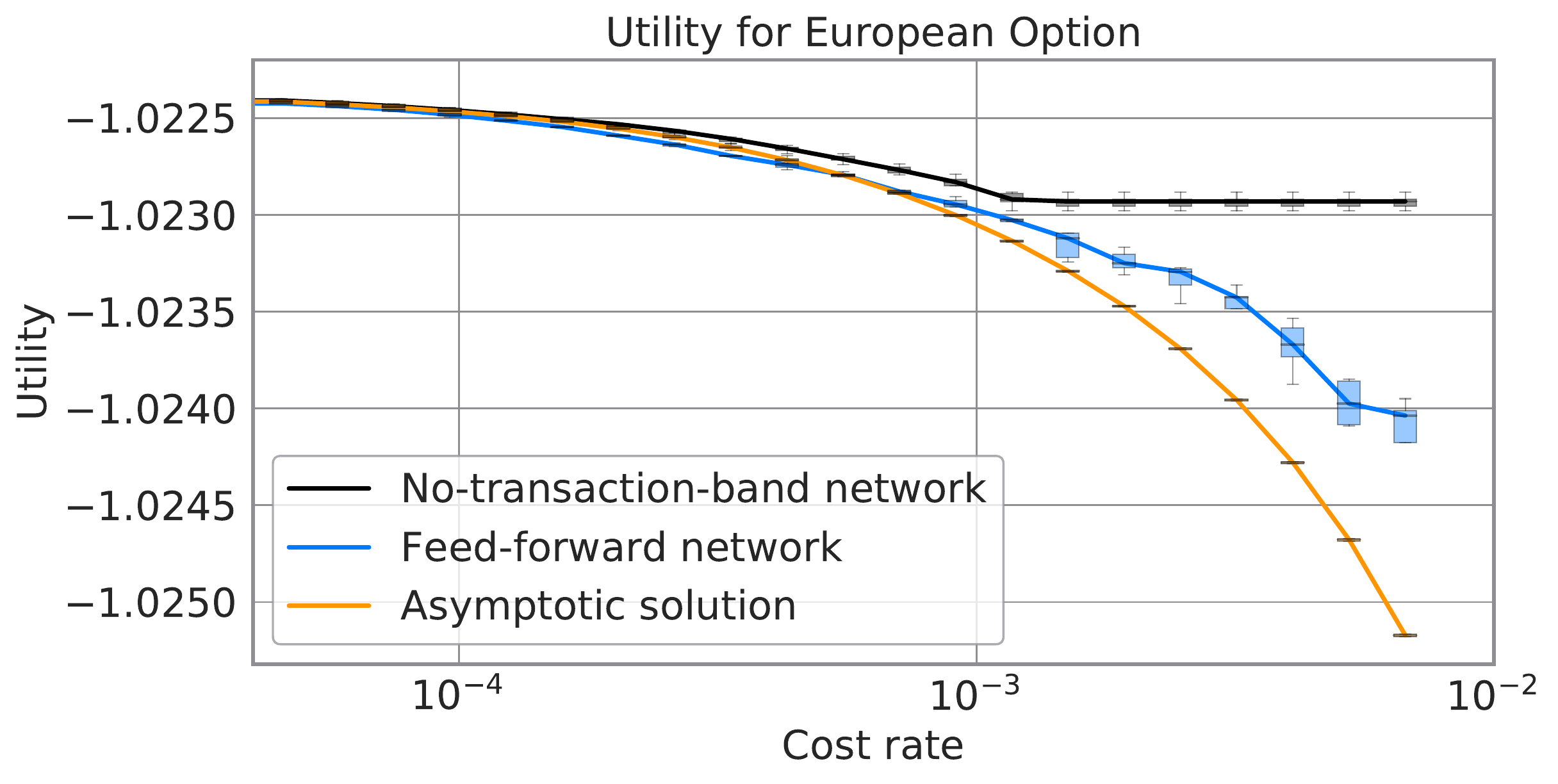}
                    }
                    \label{fig-eu-utility}
                }
                \subfigure[Price spread for European option.]{
                    \resizebox*{0.48\linewidth}{!}{
                        \includegraphics{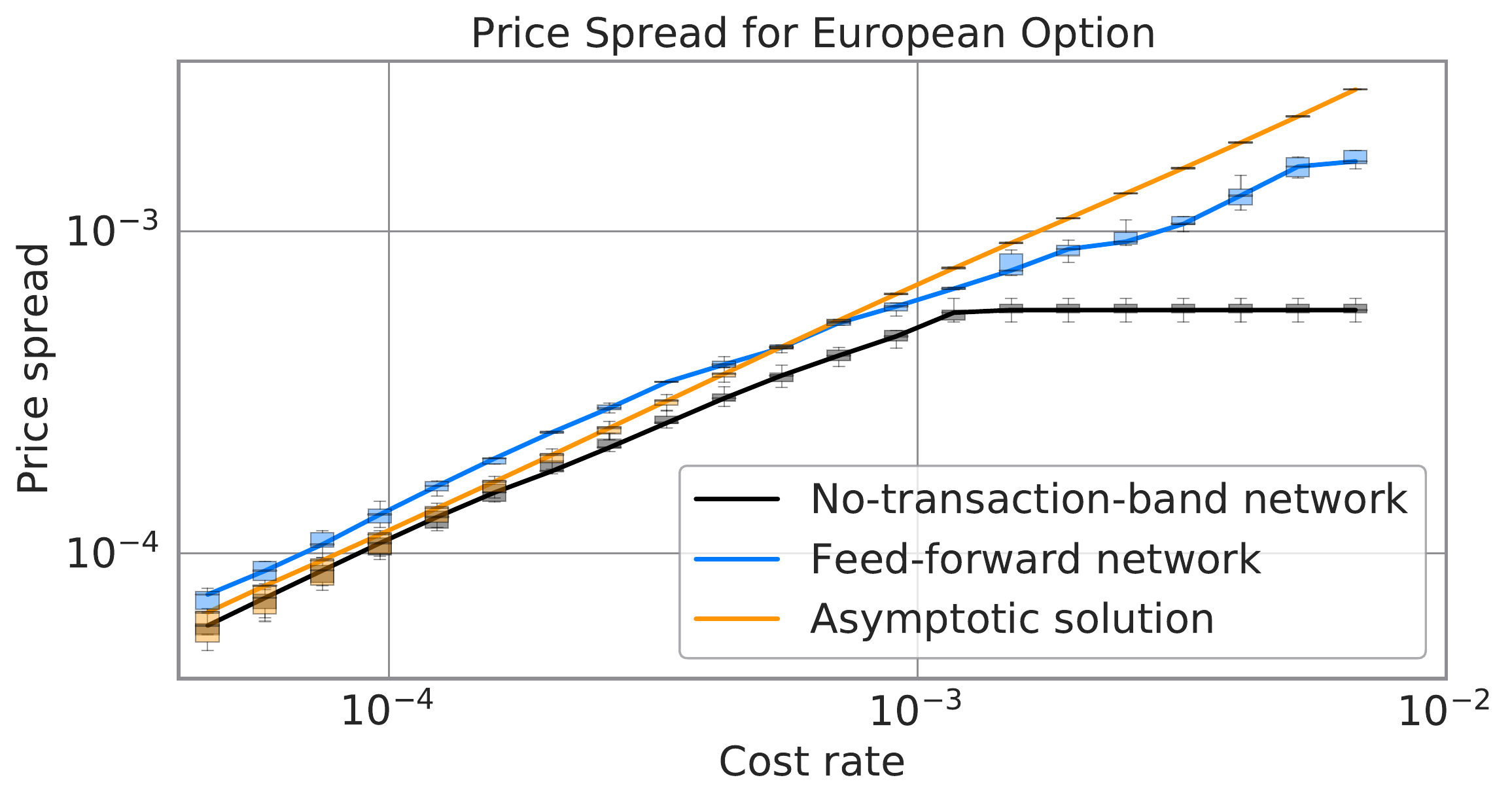}
                    }
                    \label{fig-eu-price}
                }
                \caption{
                    Final expected utility and price attained by different methods for a European option.
                    We take averages of the results from 20 times of Monte Carlo simulations.
                    The error bars are obtained by five experiments with different random seeds for initialization of a neural network and Monte Carlo paths.
                }
            \end{center}
        \end{minipage}
    \end{center}
\end{figure}
\begin{figure}
    \begin{center}
        \begin{minipage}{\linewidth}
            \begin{center}
                \subfigure[No-transaction band for small cost.]{
                    \resizebox*{0.48\linewidth}{!}{
                        \includegraphics{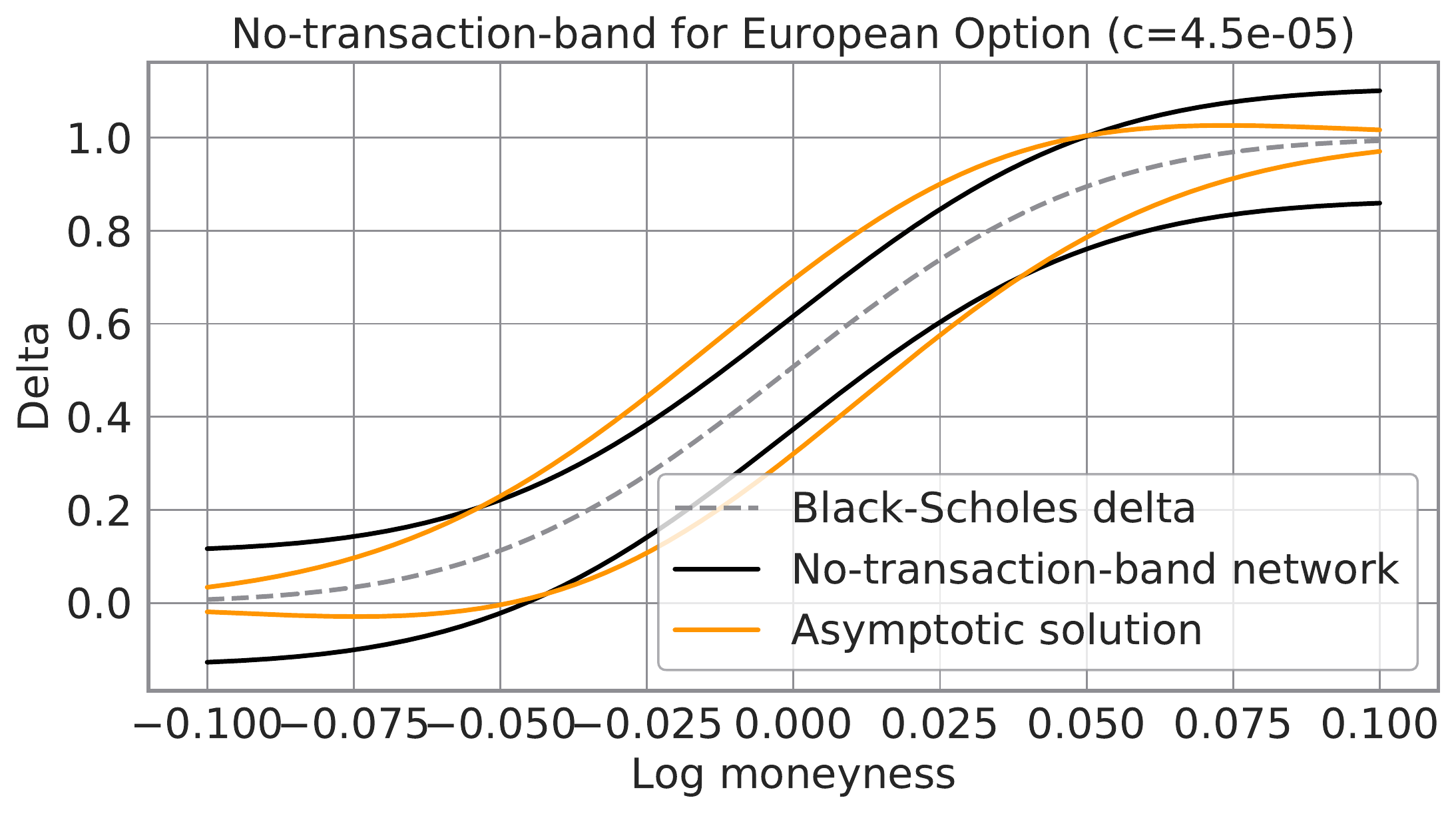}
                    }
                    \label{fig-eu-band-0}
                }
                \subfigure[No-transaction band for large cost.]{
                    \resizebox*{0.48\linewidth}{!}{
                        \includegraphics{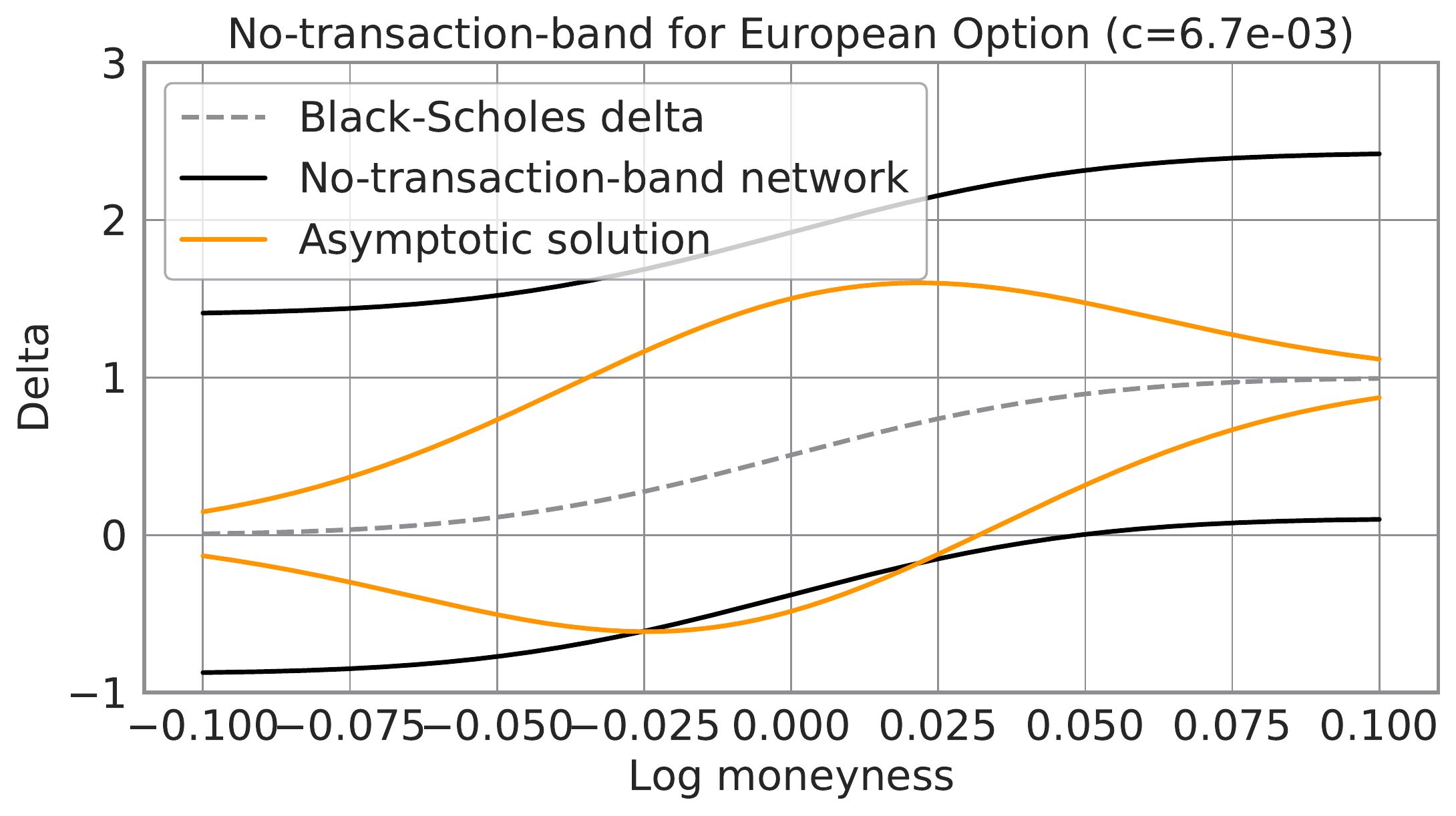}
                    }
                    \label{fig-eu-band-20}
                }
                \caption{
                    No-transaction bands obtained by the no-transaction band network and the asymptotic solution~\eqref{eq-ww} for the European option.
                    Time to maturity is $T - t_i = 15 / 365$.
                }
                \label{fig-eu-band}
            \end{center}
        \end{minipage}
    \end{center}
\end{figure}
\begin{figure}
    \begin{center}
        \begin{minipage}{\linewidth}
            \begin{center}
                \subfigure[Learning history for small cost.]{
                    \resizebox*{0.48\linewidth}{!}{\includegraphics{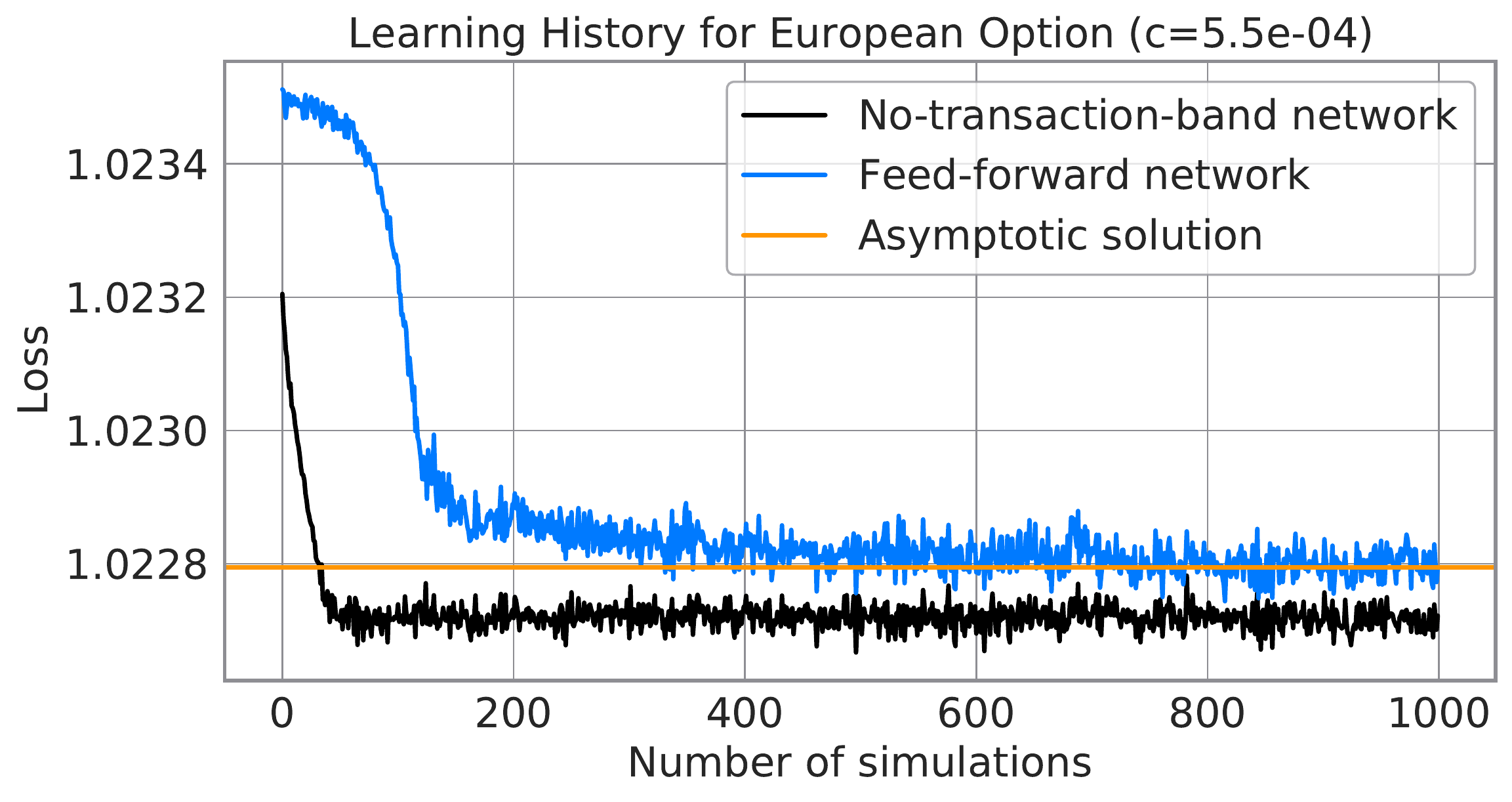}}
                    \label{fig-eu-history-10}
                }
                \subfigure[Learning history for large cost.]{
                    \resizebox*{0.48\linewidth}{!}{\includegraphics{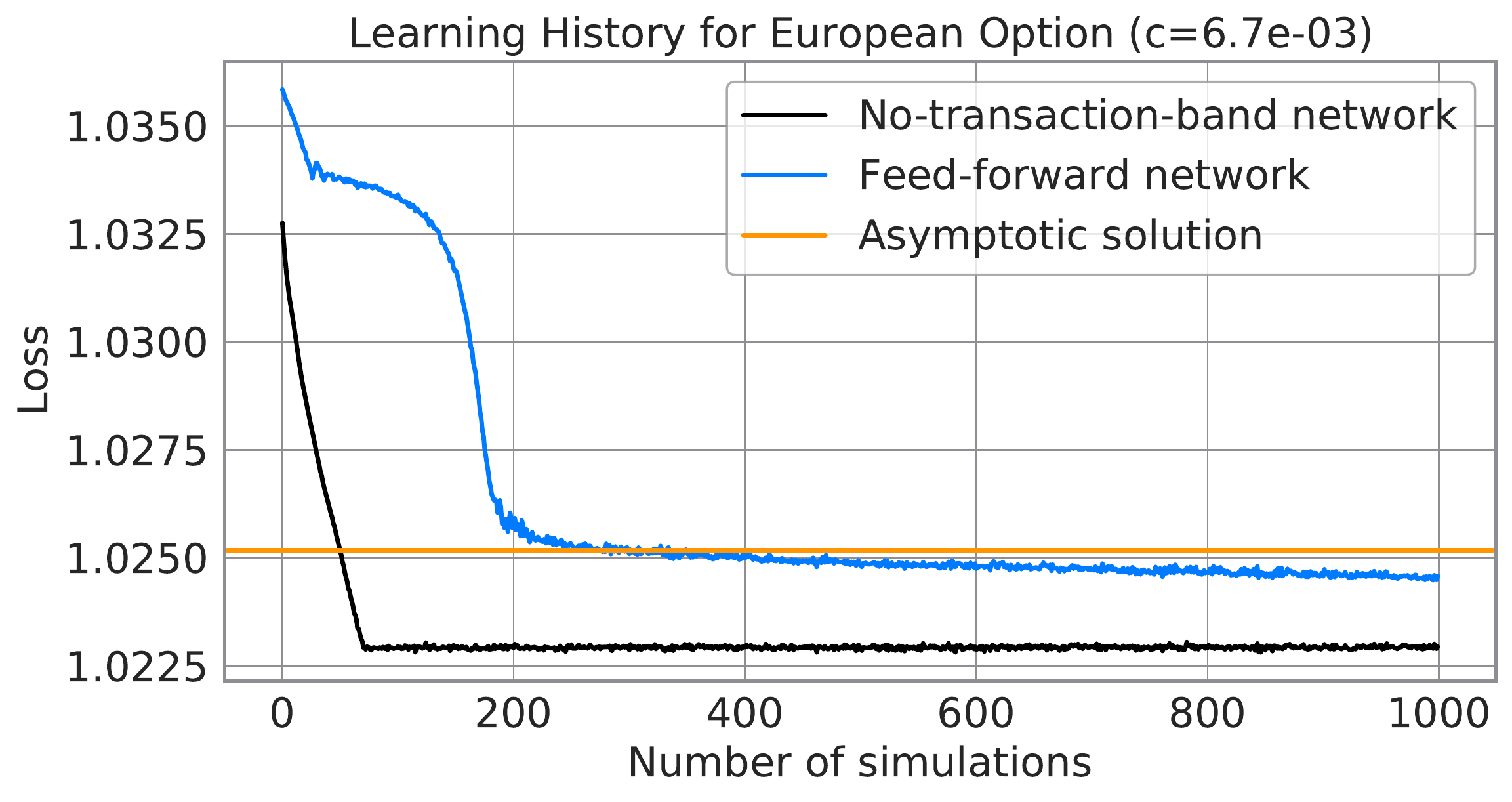}}
                    \label{fig-eu-history-20}
                }
                \caption{
                    Learning histories of different methods for a European option.
                    The horizontal axis stands for the number of Monte Carlo simulations.
                    The vertical axis shows the loss which is the negative of the expectation value of utility.
                    The validation loss is computed as an average of the results from 20 Monte Carlo simulations.
                    The losses for the asymptotic solution is plotted as constants as it does not need training.
                }
                \label{fig-eu-history}
            \end{center}
        \end{minipage}
    \end{center}
\end{figure}

\subsection{Lookback Option} \label{subsec-lb}

We consider a lookback option with a fixed strike $K = 1.03$ (See Example~\ref{example-lookback}).
A set of relevant information includes a cumulative maximum of the log-moneyness $M_{t_i} \equiv \max [\{\log(S_{t_j} / K)\}_{t_j \leq t_i}]$ on top of the features used for a European option.
We also employ the asymptotic formula~\eqref{eq-ww} using numerically computed Black--Scholes' delta and gamma of the lookback option.

The no-transaction band network achieves the highest utility as shown in Figure~\ref{fig-lb-utility}, and accordingly the cheapest prices in Figure~\ref{fig-lb-price}.
Again, the utility and price for the no-transaction band do not get worse than some value, beyond which the no-transaction band network learns to keep the naked position to save on transaction costs.
Interestingly, the scaling relation $p(c) - p(0) \propto c^{2 / 3}$ is again approximately realized for the asymptotic solution.
The no-transaction band network is close to this scaling behavior, which implies its optimality for small transaction costs.

Figure~\ref{fig-lb-history} exhibits that the no-transaction band network can be trained much quicker than the ordinary feed-forward network.
Surprisingly, the no-transaction band network achieves its optima as fast as it learns to hedge a European option, even though the lookback option bears further complication of path-dependence and needs more features.
This result suggests the scalability of our method concerning the number of input features.

\begin{figure}
    \begin{center}
        \begin{minipage}{\linewidth}
            \begin{center}
                \subfigure[Utility for lookback option.]{
                    \resizebox*{0.48\linewidth}{!}{\includegraphics{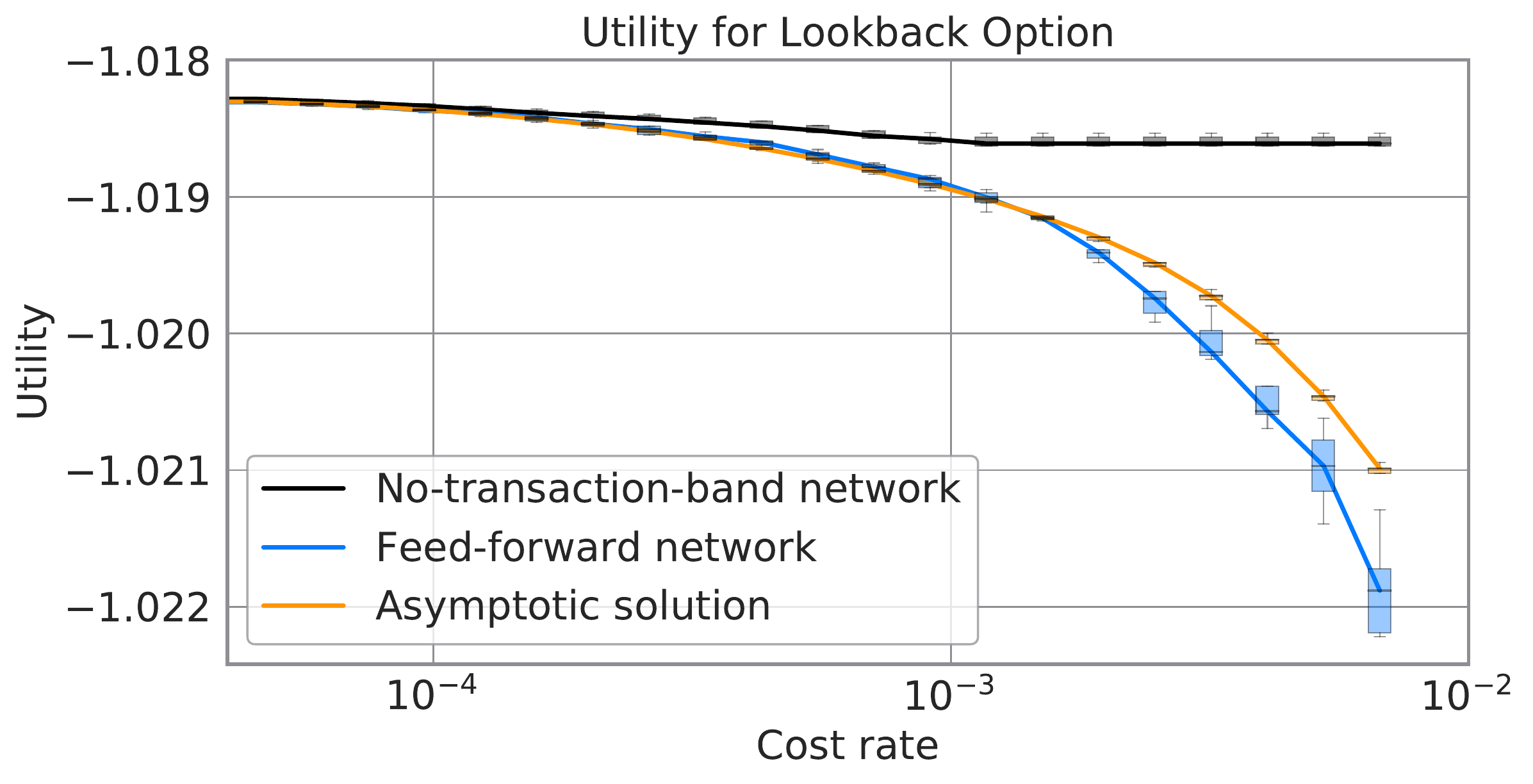}}
                    \label{fig-lb-utility}
                }
                \subfigure[Price spread for lookback option.]{
                    \resizebox*{0.48\linewidth}{!}{\includegraphics{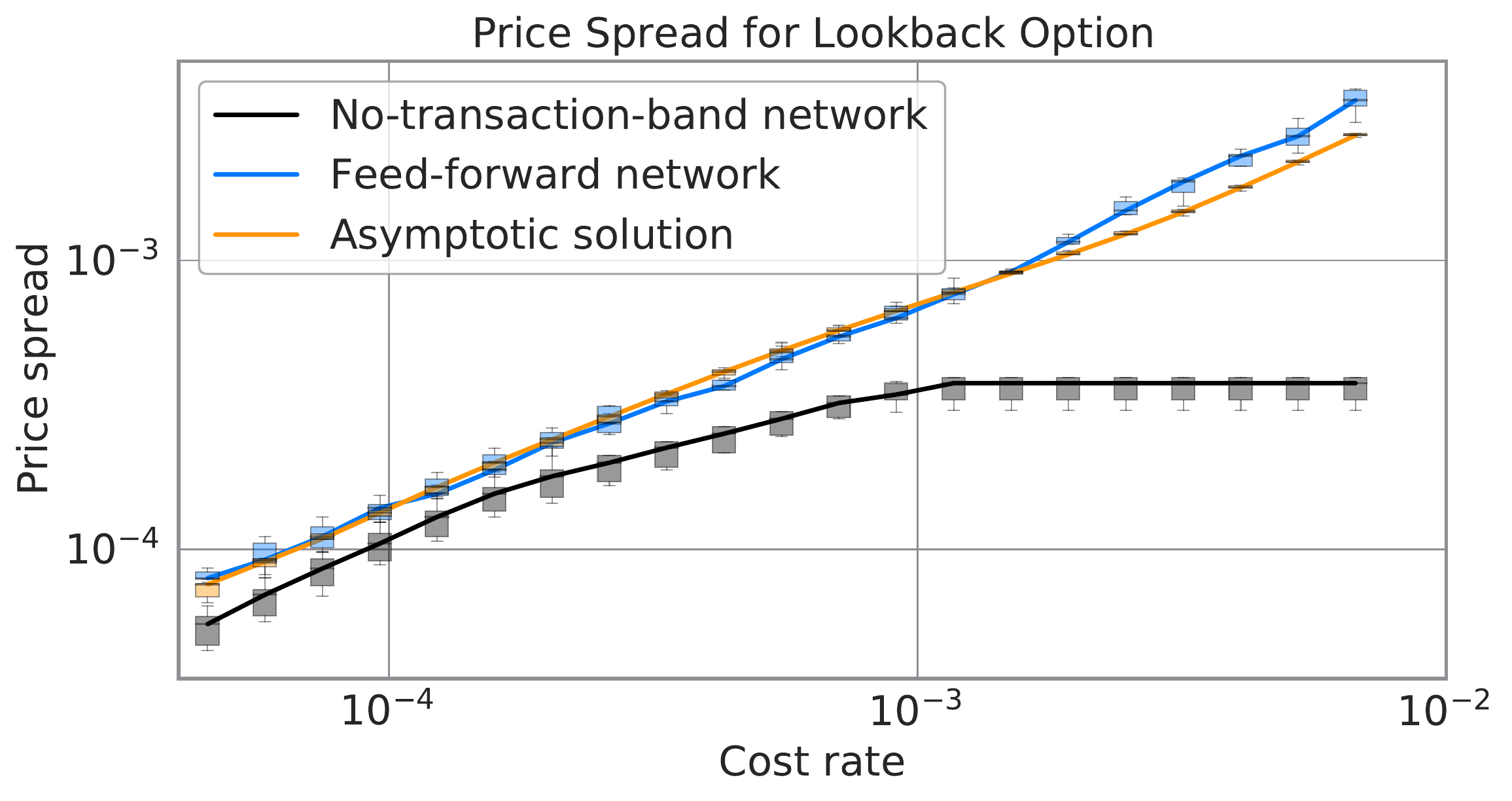}}
                    \label{fig-lb-price}
                }
                \caption{
                    Final expected utility and price attained by different methods for a lookback option.
                    We take averages of the results from 20 times of Monte Carlo simulations.
                    The error bars are obtained by five experiments with different random seeds for initialization of a neural network and Monte Carlo paths.
                }
            \end{center}
        \end{minipage}
    \end{center}
\end{figure}
\begin{figure}
    \begin{center}
        \begin{minipage}{\linewidth}
            \begin{center}
                \subfigure[Learning history for small cost.]{
                    \resizebox*{0.48\linewidth}{!}{\includegraphics{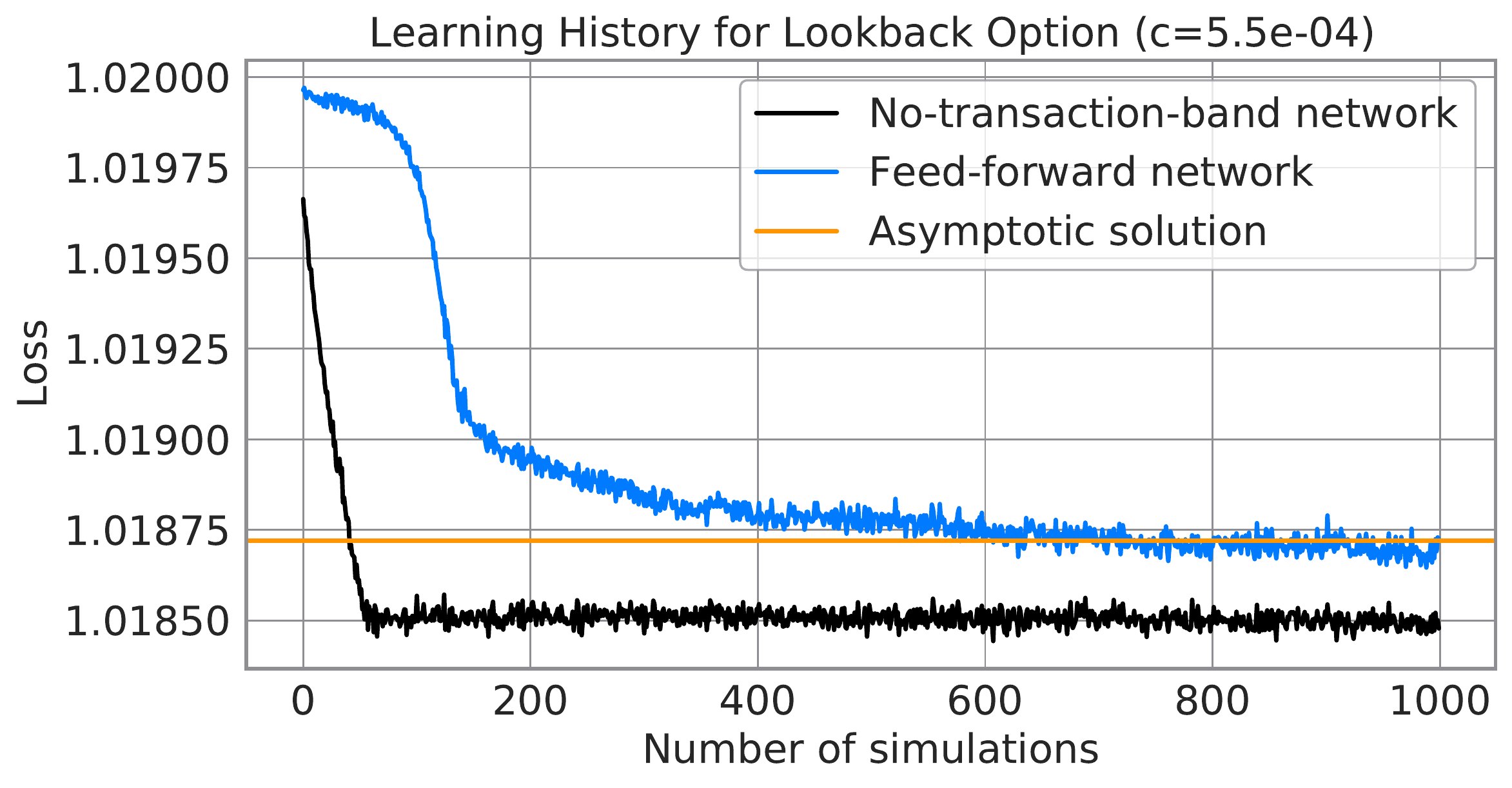}}
                    \label{fig-lb-history-10}
                }
                \subfigure[Learning history for large cost.]{
                    \resizebox*{0.48\linewidth}{!}{\includegraphics{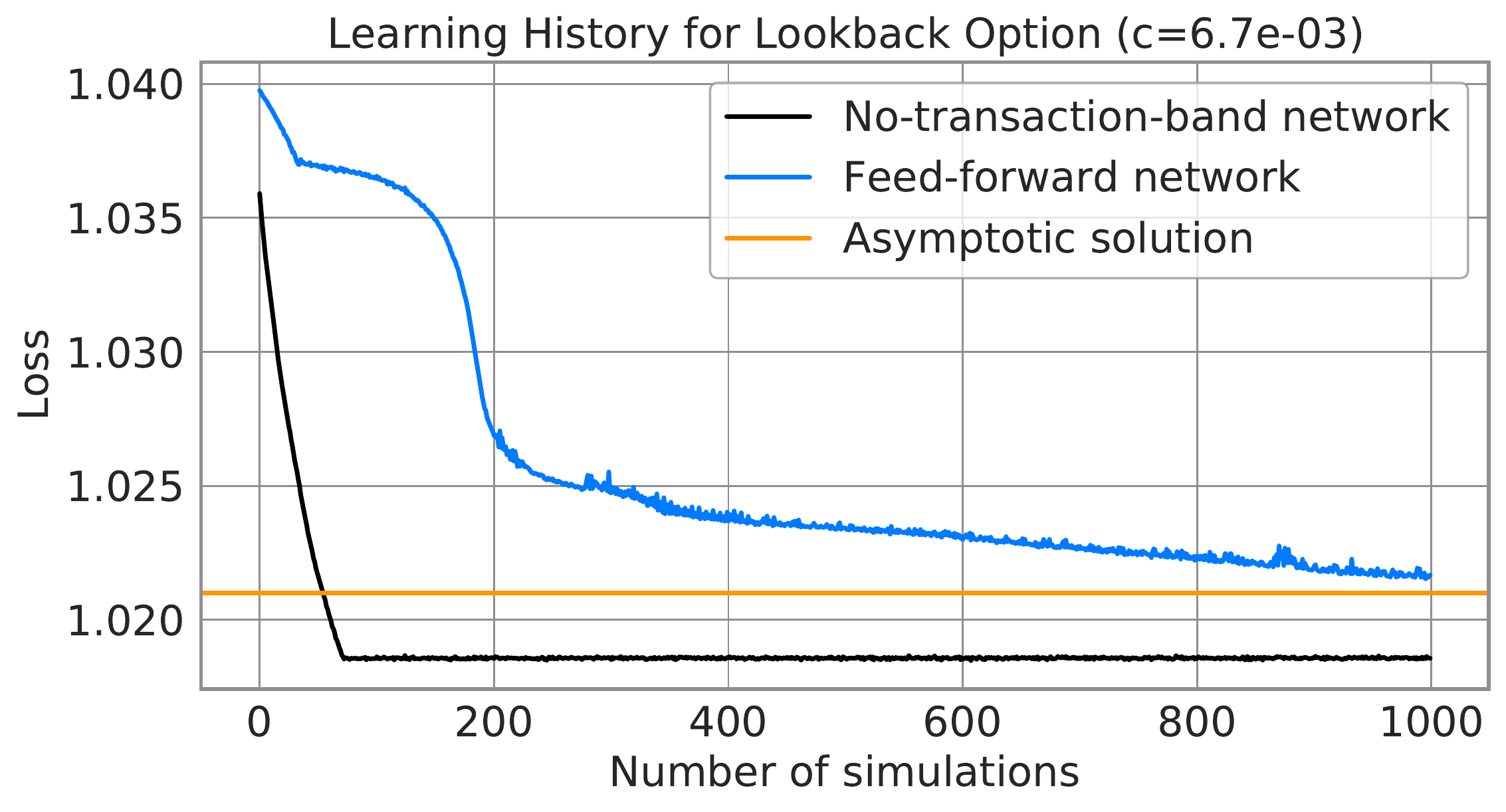}}
                    \label{fig-lb-history-20}
                }
                \caption{
                    Learning histories of different methods for a lookback option.
                    The horizontal axis stands for the number of Monte Carlo simulations and backpropagations.
                    The vertical axis shows the loss which is the negative of the expectation value of utility.
                    The validation loss is computed as an average of the results from 20 Monte Carlo simulations.
                    The losses for the asymptotic solution is plotted as constants as it does not need training.
                }
            \label{fig-lb-history}
            \end{center}
        \end{minipage}
    \end{center}
\end{figure}

\section{Conclusion} \label{sec-conclusion}

We proposed the no-transaction band network, which is a new but simple neural network architecture for efficient deep hedging.
This network properly encodes an inductive bias toward the optimal hedge, which is buttressed by the analytic consideration of the hedging optimization problem.
Our experiments demonstrated the advantage of the no-transaction band network concerning the high utility, cheaper price, and fast learning.

We expect our method to be extended to more general market dynamics, derivatives, and market frictions.
Practical applications require refined market dynamics including volatility smiles and term structures, stochastic volatility models, rough volatility~\citep{horvath2021deep}, and so forth.
Although we assumed Brownian motion for brevity, extensions to these sophisticated models are essential future works.
Moreover, our discussion assumes derivatives to have finite delta and gamma.
Relaxing this assumption helps deep hedging algorithms price more complex but practically relevant derivatives such as binary options and barrier options.
Besides, while we focused on the proportional transaction cost, different transaction cost modeling along with perpetual market impacts~\citep{PhysRevLett.122.108302} deserve further analysis.
Also, pricing under constraints such as liquidity restrictions and risk limits is another practical requirement.

Another challenge is to incorporate multiple hedging instruments.
For instance, as mentioned in \cite{buehler2019deep}, a derivative of an asset with stochastic volatility requires a volatility swap as well as the asset itself for complete hedging.
This extension requires solving a high-dimensional optimization problem and thus seems to be a challenging task.

We believe that our proposal brings deep hedging closer to practical applications.
Precise and fast computation is an essential need for derivative pricing and the no-transaction band network can meet them by adding a simple layer.
We emphasize that the neural network-based pricing is constructive because neural networks directly present hedging strategies, which enables a dealer to quote the computed price.
Further essential work is to apply our method to real market data rather than a synthetic market.

\section*{Acknowledgments}

The authors are grateful to Takuya Shimada for discussions.

\appendix

\section{Learning Histories for Different Learning Rates} \label{appendix-lr}

Figure~\ref{fig-eu-history-lr} shows that the no-transaction band network outperforms the ordinary feed-forward network for different learning rates from $0.0001$ to $0.01$.
It worth noting that a high learning rate achieves fast conversion without compromising the final expected utility.

\begin{figure}
    \begin{center}
        \begin{minipage}{\linewidth}
            \begin{center}
                \subfigure[Learning history with small learning rate for small cost.]{
                    \resizebox*{0.48\linewidth}{!}{
                        \includegraphics{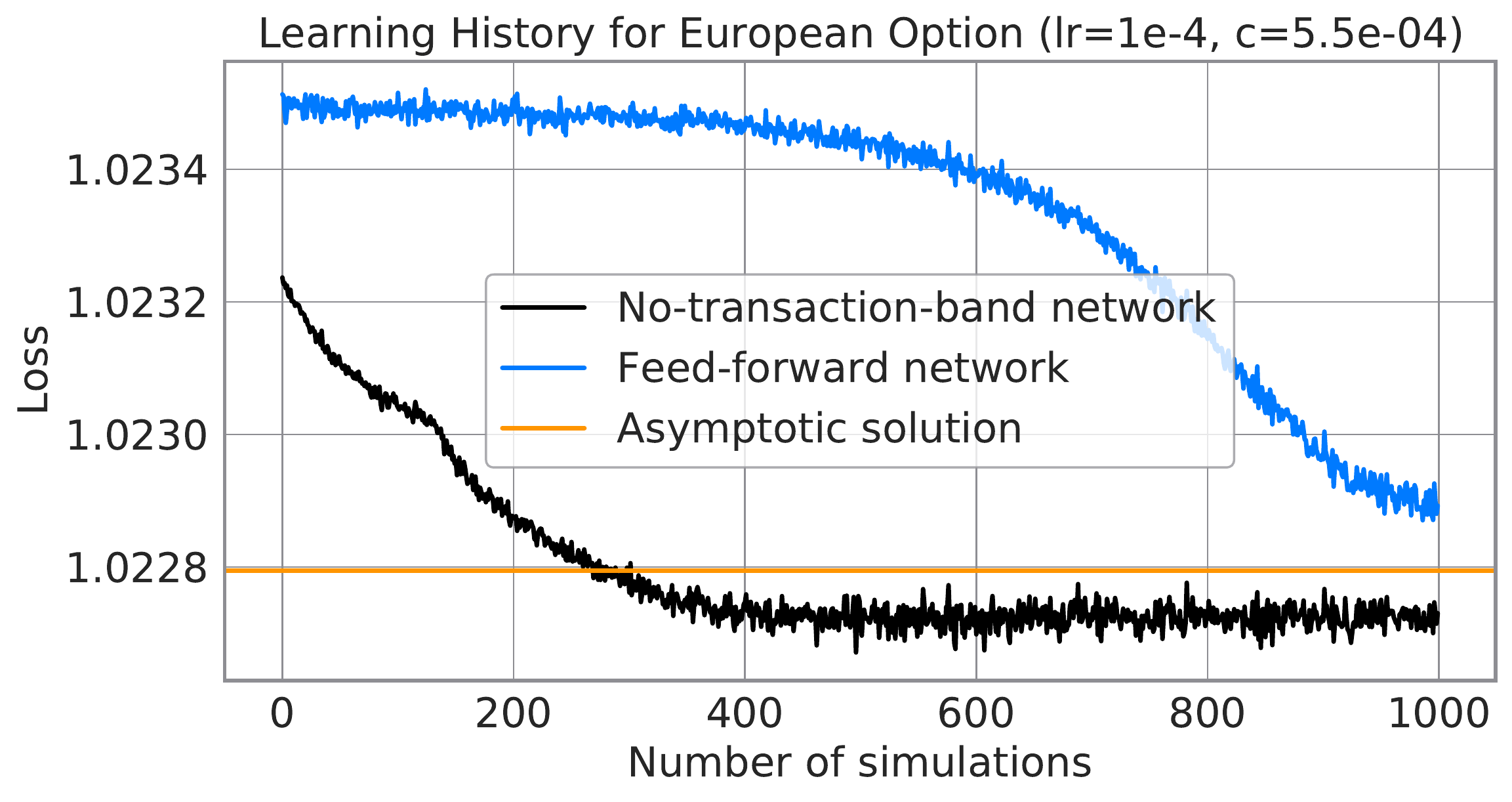}
                    }
                    \label{fig-eu-history-lr4-10}
                }
                \subfigure[Learning history with small learning rate for large cost.]{
                    \resizebox*{0.48\linewidth}{!}{
                        \includegraphics{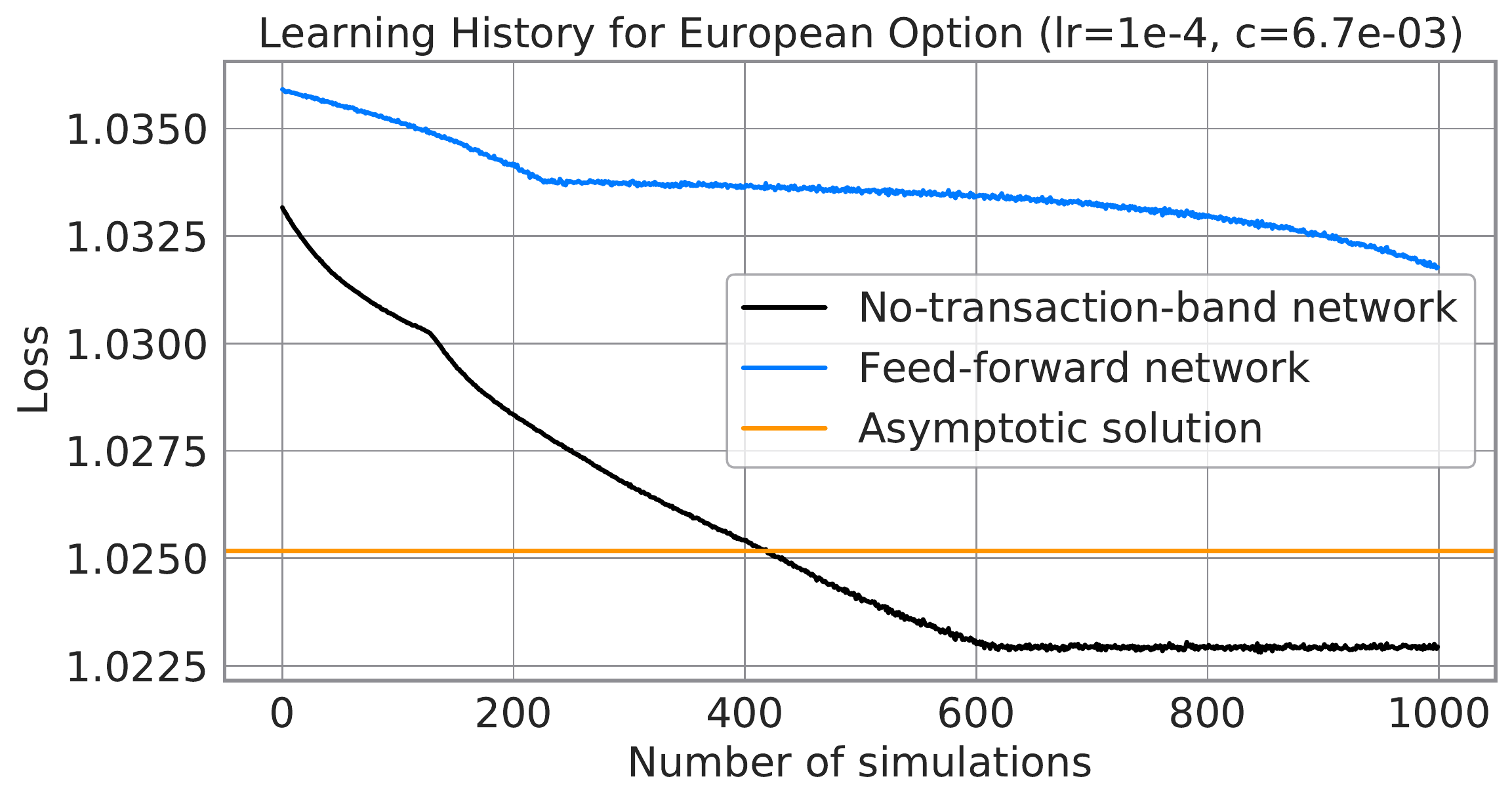}
                    }
                    \label{fig-eu-history-lr4-20}
                }
                \subfigure[Learning history with large learning rate for small cost.]{
                    \resizebox*{0.48\linewidth}{!}{
                        \includegraphics{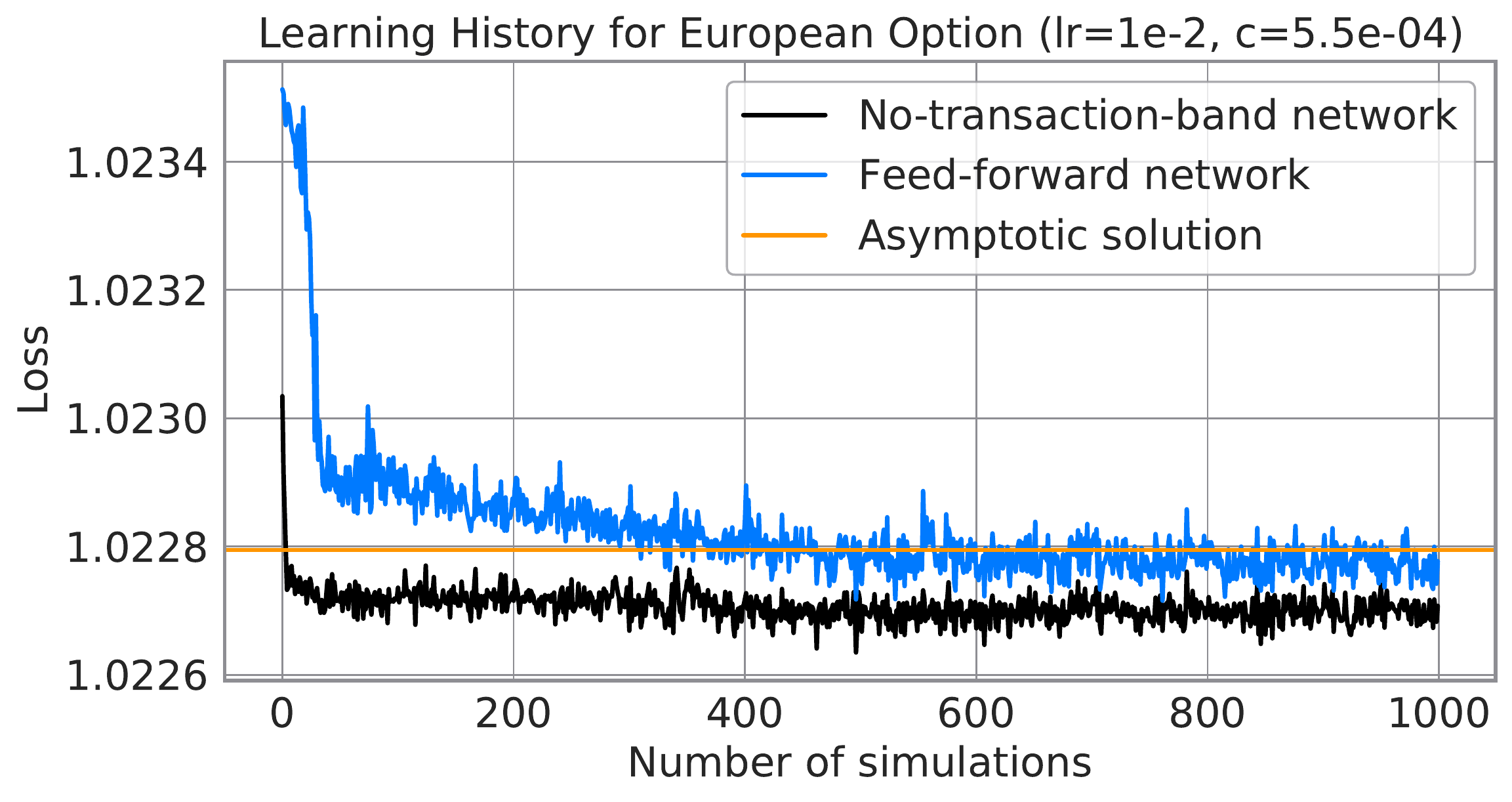}
                    }
                    \label{fig-eu-history-lr2-10}
                }
                \subfigure[Learning history with large learning rate for large cost.]{
                    \resizebox*{0.48\linewidth}{!}{
                        \includegraphics{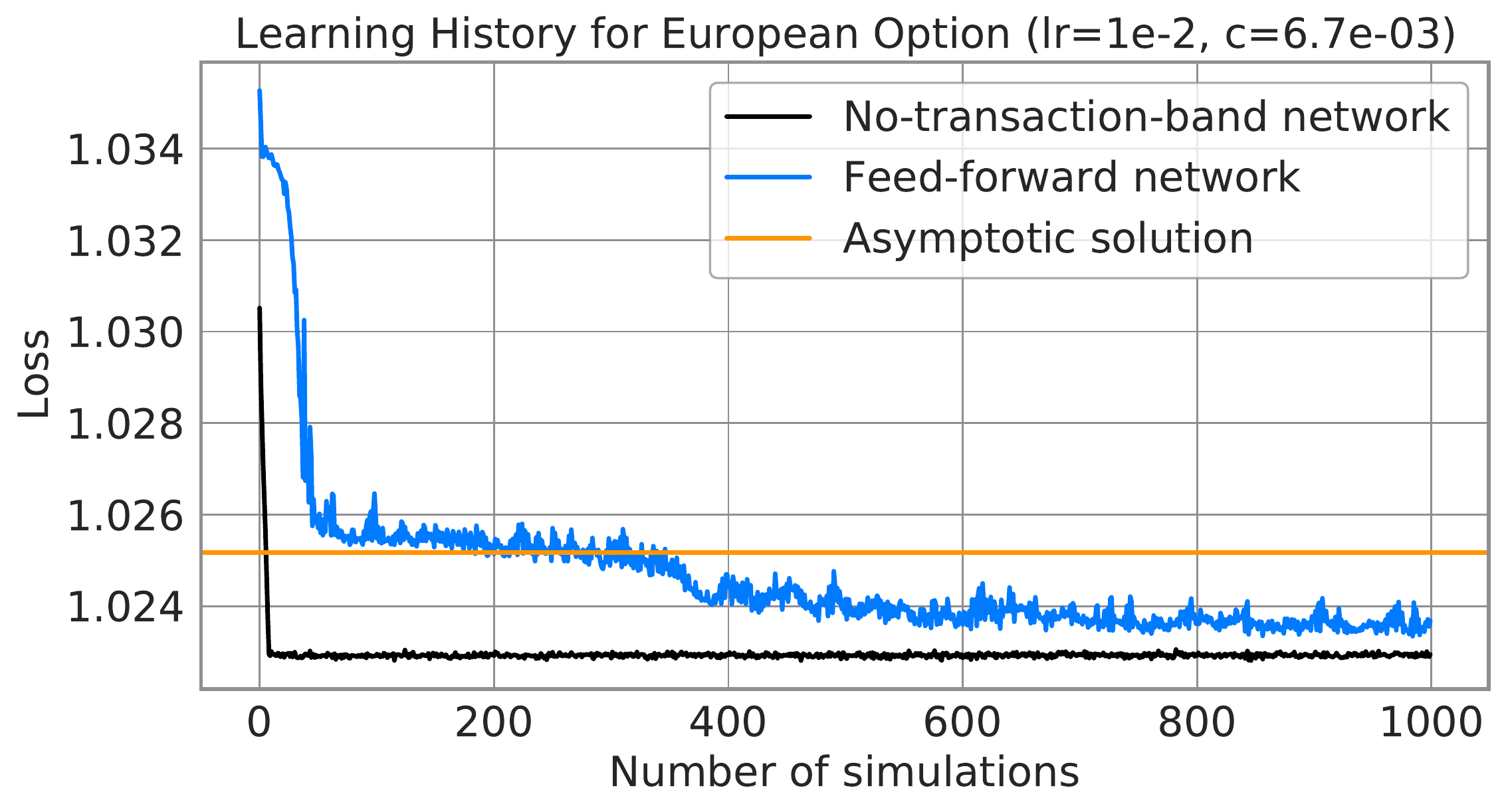}
                    }
                    \label{fig-eu-history-lr2-20}
                }
                \caption{
                    Learning histories of different methods and learning rates for European option.
                    The learning rates are $\mathsf{lr} = 0.01, 0.0001$ while the other setups are the same with the experiments in Sec.~\ref{subsec-eu} where $\mathsf{lr} = 0.001$.
                }
                \label{fig-eu-history-lr}
            \end{center}
        \end{minipage}
    \end{center}
\end{figure}

\section{Tables of Computed Utility and Prices} \label{appendix-table}

Tables~\ref{table-eu-utility}, \ref{table-eu-price}, \ref{table-lb-utility}, \ref{table-lb-price} show the computed utilities and prices for different methods and transaction costs.
\clearpage
\begin{table}
\centering
\caption{Utilities for European option.}
\label{table-eu-utility}
\begin{tabular}{rlll}
\toprule
     Cost &       No-transaction-band &              Feed-forward &                Asymptotic \\
 0.000000 &  $-1.022347 \pm 0.000004$ &  $-1.022348 \pm 0.000005$ &  $-1.022347 \pm 0.000004$ \\
 0.000045 &  $-1.022407 \pm 0.000006$ &  $-1.022422 \pm 0.000005$ &  $-1.022410 \pm 0.000008$ \\
 0.000058 &  $-1.022421 \pm 0.000006$ &  $-1.022438 \pm 0.000007$ &  $-1.022423 \pm 0.000008$ \\
 0.000075 &  $-1.022437 \pm 0.000007$ &  $-1.022459 \pm 0.000007$ &  $-1.022440 \pm 0.000008$ \\
 0.000096 &  $-1.022456 \pm 0.000007$ &  $-1.022483 \pm 0.000009$ &  $-1.022460 \pm 0.000009$ \\
 0.000123 &  $-1.022479 \pm 0.000008$ &  $-1.022512 \pm 0.000006$ &  $-1.022484 \pm 0.000009$ \\
 0.000158 &  $-1.022507 \pm 0.000010$ &  $-1.022546 \pm 0.000004$ &  $-1.022514 \pm 0.000009$ \\
 0.000203 &  $-1.022539 \pm 0.000010$ &  $-1.022590 \pm 0.000008$ &  $-1.022551 \pm 0.000010$ \\
 0.000261 &  $-1.022572 \pm 0.000011$ &  $-1.022638 \pm 0.000007$ &  $-1.022596 \pm 0.000011$ \\
 0.000335 &  $-1.022614 \pm 0.000012$ &  $-1.022697 \pm 0.000007$ &  $-1.022650 \pm 0.000012$ \\
 0.000431 &  $-1.022661 \pm 0.000015$ &  $-1.022746 \pm 0.000012$ &  $-1.022715 \pm 0.000013$ \\
 0.000553 &  $-1.022711 \pm 0.000019$ &  $-1.022799 \pm 0.000020$ &  $-1.022795 \pm 0.000014$ \\
 0.000710 &  $-1.022768 \pm 0.000021$ &  $-1.022876 \pm 0.000012$ &  $-1.022889 \pm 0.000015$ \\
 0.000912 &  $-1.022828 \pm 0.000023$ &  $-1.022953 \pm 0.000040$ &  $-1.023002 \pm 0.000017$ \\
 0.001171 &  $-1.022921 \pm 0.000035$ &  $-1.023046 \pm 0.000040$ &  $-1.023135 \pm 0.000019$ \\
 0.001503 &  $-1.022934 \pm 0.000033$ &  $-1.023155 \pm 0.000063$ &  $-1.023291 \pm 0.000021$ \\
 0.001930 &  $-1.022934 \pm 0.000033$ &  $-1.023241 \pm 0.000050$ &  $-1.023472 \pm 0.000023$ \\
 0.002479 &  $-1.022934 \pm 0.000033$ &  $-1.023334 \pm 0.000070$ &  $-1.023691 \pm 0.000025$ \\
 0.003183 &  $-1.022934 \pm 0.000033$ &  $-1.023500 \pm 0.000156$ &  $-1.023956 \pm 0.000027$ \\
 0.004087 &  $-1.022934 \pm 0.000033$ &  $-1.023679 \pm 0.000119$ &  $-1.024279 \pm 0.000029$ \\
 0.005248 &  $-1.022934 \pm 0.000033$ &  $-1.023971 \pm 0.000105$ &  $-1.024678 \pm 0.000031$ \\
 0.006738 &  $-1.022934 \pm 0.000033$ &  $-1.024159 \pm 0.000243$ &  $-1.025171 \pm 0.000033$ \\
\end{tabular}
\end{table}
\begin{table}
\centering
\caption{Prices for European option.}
\label{table-eu-price}
\begin{tabular}{rlll}
\toprule
     Cost &      No-transaction-band &             Feed-forward &               Asymptotic \\
 0.000000 &  $0.022101 \pm 0.000004$ &  $0.022101 \pm 0.000004$ &  $0.022101 \pm 0.000004$ \\
 0.000045 &  $0.022160 \pm 0.000005$ &  $0.022174 \pm 0.000005$ &  $0.022162 \pm 0.000007$ \\
 0.000058 &  $0.022173 \pm 0.000006$ &  $0.022189 \pm 0.000007$ &  $0.022176 \pm 0.000008$ \\
 0.000075 &  $0.022189 \pm 0.000007$ &  $0.022210 \pm 0.000008$ &  $0.022191 \pm 0.000008$ \\
 0.000096 &  $0.022208 \pm 0.000007$ &  $0.022234 \pm 0.000009$ &  $0.022211 \pm 0.000008$ \\
 0.000123 &  $0.022230 \pm 0.000008$ &  $0.022263 \pm 0.000007$ &  $0.022235 \pm 0.000009$ \\
 0.000158 &  $0.022257 \pm 0.000009$ &  $0.022296 \pm 0.000004$ &  $0.022265 \pm 0.000009$ \\
 0.000203 &  $0.022288 \pm 0.000010$ &  $0.022340 \pm 0.000009$ &  $0.022300 \pm 0.000010$ \\
 0.000261 &  $0.022321 \pm 0.000011$ &  $0.022386 \pm 0.000007$ &  $0.022344 \pm 0.000011$ \\
 0.000335 &  $0.022362 \pm 0.000012$ &  $0.022444 \pm 0.000006$ &  $0.022397 \pm 0.000012$ \\
 0.000431 &  $0.022408 \pm 0.000015$ &  $0.022491 \pm 0.000011$ &  $0.022461 \pm 0.000013$ \\
 0.000553 &  $0.022457 \pm 0.000019$ &  $0.022543 \pm 0.000019$ &  $0.022539 \pm 0.000014$ \\
 0.000710 &  $0.022512 \pm 0.000020$ &  $0.022618 \pm 0.000012$ &  $0.022631 \pm 0.000015$ \\
 0.000912 &  $0.022571 \pm 0.000022$ &  $0.022694 \pm 0.000039$ &  $0.022742 \pm 0.000017$ \\
 0.001171 &  $0.022662 \pm 0.000034$ &  $0.022785 \pm 0.000039$ &  $0.022872 \pm 0.000018$ \\
 0.001503 &  $0.022675 \pm 0.000032$ &  $0.022891 \pm 0.000061$ &  $0.023024 \pm 0.000020$ \\
 0.001930 &  $0.022675 \pm 0.000032$ &  $0.022976 \pm 0.000048$ &  $0.023201 \pm 0.000022$ \\
 0.002479 &  $0.022675 \pm 0.000032$ &  $0.023066 \pm 0.000068$ &  $0.023415 \pm 0.000024$ \\
 0.003183 &  $0.022675 \pm 0.000032$ &  $0.023229 \pm 0.000152$ &  $0.023673 \pm 0.000026$ \\
 0.004087 &  $0.022675 \pm 0.000032$ &  $0.023404 \pm 0.000115$ &  $0.023989 \pm 0.000028$ \\
 0.005248 &  $0.022675 \pm 0.000032$ &  $0.023688 \pm 0.000102$ &  $0.024378 \pm 0.000030$ \\
 0.006738 &  $0.022675 \pm 0.000032$ &  $0.023872 \pm 0.000237$ &  $0.024859 \pm 0.000032$ \\
\end{tabular}
\end{table}
\begin{table}
\centering
\caption{Utilities for lookback option.}
\label{table-lb-utility}
\begin{tabular}{rlll}
\toprule
     Cost &       No-transaction-band &              Feed-forward &                Asymptotic \\
 0.000000 &  $-1.022347 \pm 0.000004$ &  $-1.022348 \pm 0.000005$ &  $-1.022347 \pm 0.000004$ \\
 0.000045 &  $-1.022407 \pm 0.000006$ &  $-1.022422 \pm 0.000005$ &  $-1.022410 \pm 0.000008$ \\
 0.000058 &  $-1.022421 \pm 0.000006$ &  $-1.022438 \pm 0.000007$ &  $-1.022423 \pm 0.000008$ \\
 0.000075 &  $-1.022437 \pm 0.000007$ &  $-1.022459 \pm 0.000007$ &  $-1.022440 \pm 0.000008$ \\
 0.000096 &  $-1.022456 \pm 0.000007$ &  $-1.022483 \pm 0.000009$ &  $-1.022460 \pm 0.000009$ \\
 0.000123 &  $-1.022479 \pm 0.000008$ &  $-1.022512 \pm 0.000006$ &  $-1.022484 \pm 0.000009$ \\
 0.000158 &  $-1.022507 \pm 0.000010$ &  $-1.022546 \pm 0.000004$ &  $-1.022514 \pm 0.000009$ \\
 0.000203 &  $-1.022539 \pm 0.000010$ &  $-1.022590 \pm 0.000008$ &  $-1.022551 \pm 0.000010$ \\
 0.000261 &  $-1.022572 \pm 0.000011$ &  $-1.022638 \pm 0.000007$ &  $-1.022596 \pm 0.000011$ \\
 0.000335 &  $-1.022614 \pm 0.000012$ &  $-1.022697 \pm 0.000007$ &  $-1.022650 \pm 0.000012$ \\
 0.000431 &  $-1.022661 \pm 0.000015$ &  $-1.022746 \pm 0.000012$ &  $-1.022715 \pm 0.000013$ \\
 0.000553 &  $-1.022711 \pm 0.000019$ &  $-1.022799 \pm 0.000020$ &  $-1.022795 \pm 0.000014$ \\
 0.000710 &  $-1.022768 \pm 0.000021$ &  $-1.022876 \pm 0.000012$ &  $-1.022889 \pm 0.000015$ \\
 0.000912 &  $-1.022828 \pm 0.000023$ &  $-1.022953 \pm 0.000040$ &  $-1.023002 \pm 0.000017$ \\
 0.001171 &  $-1.022921 \pm 0.000035$ &  $-1.023046 \pm 0.000040$ &  $-1.023135 \pm 0.000019$ \\
 0.001503 &  $-1.022934 \pm 0.000033$ &  $-1.023155 \pm 0.000063$ &  $-1.023291 \pm 0.000021$ \\
 0.001930 &  $-1.022934 \pm 0.000033$ &  $-1.023241 \pm 0.000050$ &  $-1.023472 \pm 0.000023$ \\
 0.002479 &  $-1.022934 \pm 0.000033$ &  $-1.023334 \pm 0.000070$ &  $-1.023691 \pm 0.000025$ \\
 0.003183 &  $-1.022934 \pm 0.000033$ &  $-1.023500 \pm 0.000156$ &  $-1.023956 \pm 0.000027$ \\
 0.004087 &  $-1.022934 \pm 0.000033$ &  $-1.023679 \pm 0.000119$ &  $-1.024279 \pm 0.000029$ \\
 0.005248 &  $-1.022934 \pm 0.000033$ &  $-1.023971 \pm 0.000105$ &  $-1.024678 \pm 0.000031$ \\
 0.006738 &  $-1.022934 \pm 0.000033$ &  $-1.024159 \pm 0.000243$ &  $-1.025171 \pm 0.000033$ \\
\end{tabular}
\end{table}
\begin{table}
\centering
\caption{Prices for lookback option.}
\label{table-lb-price}
\begin{tabular}{rlll}
\toprule
     Cost &      No-transaction-band &             Feed-forward &               Asymptotic \\
 0.000000 &  $0.018040 \pm 0.000005$ &  $0.018035 \pm 0.000005$ &  $0.018053 \pm 0.000008$ \\
 0.000045 &  $0.018115 \pm 0.000007$ &  $0.018139 \pm 0.000008$ &  $0.018133 \pm 0.000005$ \\
 0.000058 &  $0.018128 \pm 0.000009$ &  $0.018156 \pm 0.000011$ &  $0.018150 \pm 0.000004$ \\
 0.000075 &  $0.018145 \pm 0.000011$ &  $0.018173 \pm 0.000012$ &  $0.018170 \pm 0.000004$ \\
 0.000096 &  $0.018165 \pm 0.000013$ &  $0.018198 \pm 0.000011$ &  $0.018193 \pm 0.000005$ \\
 0.000123 &  $0.018187 \pm 0.000016$ &  $0.018224 \pm 0.000014$ &  $0.018222 \pm 0.000006$ \\
 0.000158 &  $0.018215 \pm 0.000021$ &  $0.018258 \pm 0.000018$ &  $0.018256 \pm 0.000007$ \\
 0.000203 &  $0.018240 \pm 0.000032$ &  $0.018298 \pm 0.000019$ &  $0.018297 \pm 0.000008$ \\
 0.000261 &  $0.018265 \pm 0.000038$ &  $0.018341 \pm 0.000028$ &  $0.018346 \pm 0.000009$ \\
 0.000335 &  $0.018294 \pm 0.000049$ &  $0.018393 \pm 0.000031$ &  $0.018403 \pm 0.000010$ \\
 0.000431 &  $0.018328 \pm 0.000063$ &  $0.018446 \pm 0.000038$ &  $0.018470 \pm 0.000012$ \\
 0.000553 &  $0.018373 \pm 0.000087$ &  $0.018525 \pm 0.000034$ &  $0.018547 \pm 0.000013$ \\
 0.000710 &  $0.018416 \pm 0.000097$ &  $0.018604 \pm 0.000021$ &  $0.018635 \pm 0.000015$ \\
 0.000912 &  $0.018406 \pm 0.000031$ &  $0.018715 \pm 0.000042$ &  $0.018731 \pm 0.000018$ \\
 0.001171 &  $0.018482 \pm 0.000147$ &  $0.018834 \pm 0.000056$ &  $0.018839 \pm 0.000020$ \\
 0.001503 &  $0.018523 \pm 0.000228$ &  $0.018967 \pm 0.000046$ &  $0.018964 \pm 0.000022$ \\
 0.001930 &  $0.018587 \pm 0.000354$ &  $0.019210 \pm 0.000076$ &  $0.019113 \pm 0.000023$ \\
 0.002479 &  $0.018660 \pm 0.000500$ &  $0.019529 \pm 0.000176$ &  $0.019296 \pm 0.000025$ \\
 0.003183 &  $0.018721 \pm 0.000621$ &  $0.019853 \pm 0.000144$ &  $0.019533 \pm 0.000027$ \\
 0.004087 &  $0.018800 \pm 0.000780$ &  $0.020239 \pm 0.000244$ &  $0.019849 \pm 0.000028$ \\
 0.005248 &  $0.018915 \pm 0.001010$ &  $0.020766 \pm 0.000267$ &  $0.020256 \pm 0.000029$ \\
 0.006738 &  $0.019124 \pm 0.001427$ &  $0.021625 \pm 0.000334$ &  $0.020775 \pm 0.000029$ \\
\end{tabular}
\end{table}

\clearpage

\bibliographystyle{rQUF}
\bibliography{paper}
\end{document}